\documentclass[11pt]{article}
\usepackage{amssymb,amsmath}
\usepackage{epsfig}
\usepackage{palatino}
\usepackage{graphicx}
\usepackage{textcomp}
\usepackage{hyperref}

\newtheorem{theorem}{Theorem}

\newtheorem{definition}{Definition}

\newtheorem{claim}{Claim}
\newtheorem{lemma}{Lemma}
\newtheorem{conjecture}{Conjecture}

\newtheorem{corollary}{Corollary}
\newtheorem{fact}{Fact}

 \newcommand{\qedsymb}{\hfill{\rule{2mm}{2mm}}}  
 \newenvironment{proof}[1][]{\begin{trivlist}  
 \item[\hspace{\labelsep}{\bf\noindent Proof#1:\/}] 
 }{\qedsymb\end{trivlist}}

\newcommand{\be}{\begin{eqnarray}}
\newcommand{\ee}{\end{eqnarray}}

\newcommand\floor[1]{{\lfloor #1 \rfloor}}
\newcommand\ceil[1]{{\lceil #1 \rceil}}

\newcommand\ket[1]{{ |{#1} \rangle }}

\newcommand{\ignore}[1]{}

\newcommand{\eps}{\varepsilon}
\renewcommand{\epsilon}{\varepsilon}

\title{Commuting Local Hamiltonians on Expanders, Locally Testable 
Quantum codes, and the qPCP conjecture} 
\begin{document}

\author{Dorit Aharonov\thanks{School of Computer Science and
Engineering, The Hebrew University,
Jerusalem, Israel}
 \and Lior Eldar\thanks{School of Computer Science and
Engineering, The Hebrew University,
Jerusalem, Israel.}}

\date{\today}

\maketitle


\begin{abstract}
Understanding commuting local Hamiltonians (CLHs) 
is at the heart of many questions in quantum computational complexity 
and quantum physics: quantum error correcting codes, quantum NP, 
the PCP conjecture, topological order and more.

We first consider the complexity of the problem of approximating 
the ground value of a $k$-local CLH; 
The complexity of the exact case 
remains unclear despite significant effort \cite{Bra, Aha, Sch, Has}, and  
the approximation case is tightly related to the 
major conjecture of quantum PCP. 
We show that if the underlying interacting 
graph of the CLH instance has small-set expansion which 
is $\epsilon$ close to optimal, then 
the approximation problem to within a factor $O(\epsilon)$ lies in NP. 
Thus, the better the expander, the easier it is to approximate.  
This puts a bound on the complexity of 
CLH on small-set expanders, and indicates that known 
PCP constructions such as Dinur's \cite{Din}, 
whose output instances are small-set expanders, 
cannot be used to prove a PCP theorem for CLH.

We next proceed to study the much related topic of locally testable  
quantum codes based on local stabilizers; these are ground states 
of CLHs. We show that the better the small-set expansion of the 
interaction graph underlying the stabilizer becomes, the 
{\it less} robust the code becomes.
This phenomenon seems to be inherently quantum. 
We combine the above result with an upper bound we prove on the robustness 
of stabilizer codes with {\it bad} expansion (which is a classical phenomenon)  
to derive a bound on the robustness of quantum LTCs with an {\it arbitrary} 
underlying interaction graphs. 
We derive non-trivial bounds on the robustness of quantum LTC codes; 
To the best of our knowledge the 
quantum LTCs were not studied before. 

ֿ\ignore{Finally, we initiate the study of quantum PCPs of proximity. 
In the classical theory of PCPs, PCPPs play a crucial role, 
and most known PCPs use PCPs of proximity (PCPP) \cite{BGHSV}. 
Moreover, there is a standard construction of 
an LTC based on a given PCPP \cite{BGHSV}, 
in which the robustness of the code is tightly related to the soundness 
of the PCPP; We show that the classical construction of LTC from PCPPs 
carries over to the quantum setting, 
where the type of constraints in the PCPP is translated to the  
type of constraints in the resulting LTC. 
Thus, bounds on LTCs translate to limitations on PCPPs of similar constraint 
types. Our bounds on stabilizer LTCs provide bounds on very limited 
qPCPPs; However the connection motivates further study of LTC with 
more general constraint that might shed light on the 
power of quantum PCPPs.}
Finally, we initiate the study of the quantum 
analogue of PCPs of proximity \cite{BGHSV}
which play a central 
role in classical PCPs; it is known that PCPPs induce LTCs with related 
parameters \cite{BGHSV}, and we show this holds also in the 
quantum world; this connection between LTCs and PCPs
further motivates the study of LTCs and their possible parameters. 
We defer the exposition of this part of our work to the next version of the paper. 

Much is left for further research, and in particular improving the range or applicability 
and generality of our results. 
Nevertheless we believe that the results highlight interesting limitations 
on quantum PCP and that they point at possible routes that might 
shed light on the stubborn PCP conjecture. 
\end{abstract}

\section{Introduction} 
This paper is concerned with commuting local Hamiltonians (CLHs).  
A $k$-local Hamiltonian is a Hermitian matrix $H = \sum_i H_i$ 
operating on the Hilbert space 
of $n$ $d$-dimensional particles, where each term $H_i$ 
(sometimes called constraint)
acts non-trivially 
on at most $k$ particles. 
For the Hamiltonian to be a $k$-local {\it commuting} Hamiltonian 
we require in addition that 
every two terms $H_i, H_j$ commute; W.L.O.G we assume in this case that 
$H_i$'s are projections. 
We denote the family of such Hamiltonians by $CLH(k,d)$.  
Commutative local Hamiltonians have been the subject of very intense 
research for many years in physics;  
As it turns out, CLHs are related to some of 
the deepest questions 
in quantum computational complexity as well. 
Let us first provide the background and context for the questions we 
ask here regarding CLHs. 

\subsection{Background and Context} 

\noindent{\bf The computational complexity of the Local Hamiltonian problem} 
Kitaev first defined in 1998 
the local Hamiltonian (LH) problem, where 
one is given 
a local Hamiltonian on $n$ 
qudits and two real numbers $a,b$, with $b-a\ge 1/poly(n)$,  
and is asked 
whether the lowest eigenvalue is at most $a$ or larger than $b$; 
He showed \cite{Kit} that this problem
is complete for the class $QMA$, the quantum analogue of  
NP, providing a quantum counterpart of 
the celebrated Cook-Levin theorem.

In 2003, Bravyi and Vyalyi \cite{Bra} considered 
the similar decision problem for the commuting case; 
The commuting restriction might seem at first sight to devoid the 
LH problem of its quantum nature, since all terms can be diagonalized together, 
and moreover, one usually attributes the interesting features of quantum 
mechanics to its non-commutative nature, 
cf. the Heizenberg uncertainty principle; 
following this one might suspect that the commuting 
version trivially belongs to NP. 
This intuition however is misleading, since  
ground states of CLHs can exhibit 
intricate entanglement phenomenon, such as the topological order
exhibited by the Toric code \cite{Kit2}. Using clever 
applications of representations of $C^*$-algebras
Bravyi and Vyalyi \cite{Bra} 
gave a proof that the two-local commuting case indeed lies within NP, 
and this was generalized later to more general Hamiltonian classes 
\cite{Aha, Has, Sch};  
However the complexity of the general CLH problem 
remains unresolved. 

The resolution of the complexity of the
CLH problem seems important on its own right, given the importance of 
commuting Hamiltonians in physics, its intriguing nature, its relation to 
multiparticle entanglement and 
its connection to topological order.  
But apart from these reasons, it 
is also related to the major open problem of whether 
a quantum analogue of the celebrated PCP (Probabilistically Checkable Proof) 
theorem holds; we will discuss this slightly later. 

{~}

\noindent{\bf CLH and quantum Error correcting codes} 
Arguably the most important class of quantum error correcting codes 
are the stabilizer codes (see \cite{Got}); A stabilizer code is defined 
by a set of commuting generators, each generator being
an element of the (generalized) Pauli group on 
$n$ $d$-dimensional particles, i.e. a tensor product
of $n$ Pauli operators. 
The code is defined to be the simultaneous 
$1$ eigenspace of all generators (therefore these generators are 
called {\it stabilizers}); 
if all stabilizers are $k$-local (namely, 
only $k$ out of the $n$ Pauli's are non-identity), this can 
easily be seen to be equivalent to a ground space of a $CLH(k,d)$ instance. 
Thus, CLHs and quantum stabilizer codes are intimately connected. 
A particularly beautiful example is that of the Toric code of 
Kitaev \cite{Kit2} mentioned before, due to the non-local entanglement
it exhibits;  
Various other stabilizer codes exhibiting topological order 
were defined (see e.g. quantum double models \cite{Kit4}
and many other stabilizer codes were constructed and studied 
(eg, \cite{Got, Steane1, Steane2, Calderbank, Smolin, Schlingemann,Zemor2}).
Many questions regarding stabilizer codes are of interest with much left 
to be understood:   
what are the possible parameters (distance, rate, etc.) of such codes,
and what are the tradeoffs between those parameters 
\cite{TerhalTradeoffs, Fetaya, Zemor3},  
how efficient can the description of the states in those codes be
\cite{Vidal}, what are the various properties 
of such codes with respect to perturbations \cite{Bra2, Bra3}
and thermal fluctuations \cite{Yoshida}, 
their behavior in the presence of random noise 
\cite{Preskill, Kovalev1}, all the above when constraints are 
restricted to be local, namely, the quantum analogue of LDPC codes 
(\cite{Zemor2, Zemor3,Kovalev1,Kovalev2} 
and references therein), 
and more. All those questions can be viewed as questions on CLHs, 
and their resolution is tightly related 
to our understanding of multiparticle entanglement, and of course 
to our understanding of the capabilities of quantum error correction. 

{~}

\noindent{\bf Quantum PCP} 
Both the above topics are related to what is now considered 
a major challenge in quantum complexity theory: 
the quantum PCP conjecture. 
The classical PCP theorem \cite{Arora, ALMSS, Din}
is arguably the most important 
discovery in classical theoretical computer science over the past two decades; 
it states that there is a polynomial time reduction that maps
a given Constraint Satisfaction Problem (CSP) instance 
(i.e., a collection of $k$-local constraints on $n$ 
Boolean or $d$-state variables) into another CSP instance, 
such that if the original instance is satisfiable, then so is 
its image, but if the original instance is not satisfiable, 
then no assignment satisfies more than some constant $c<1$ fraction 
of constraints in the image instance.
This implies that it is possible to check whether an assignment satisfies 
a given CSP or not, by reading only a constant number of random locations 
in the assignment (!).  
A major open problem is whether some quantum version of the PCP  
theorem holds \cite{Aha2, Aar}.    
A quantum PCP conjecture can be phrased as follows:
\begin{conjecture} {\bf Quantum PCP (qPCP)}
There exist constants $c>0,k$, and a (quantum) poly-time 
algorithm, which takes an instance of a $k$-local Hamiltonian $H$, 
to another instance of a $k$-local Hamiltonian $H'$, 
such that if $H$ is satisfiable (namely, has ground energy $0$), then so 
is $H'$, but if $H$ is unsatisfiable, the minimal eigenvalue of $H'$ is 
at least $c ||H'||$.
\end{conjecture}
A different way to state the theorem is to say that even the approximation 
of the ground energy of the given Hamiltonian to within a constant times 
its norm, is quantum NP hard; an equivalent formulation (see \cite{Aha2} 
for exact statement and a proof) is that $QMA$-hard problems have a formulation 
in which the quantum witness can be tested by measuring 
only a constant number of random qudits.  
 
Attempts to extend the classical proofs of the 
PCP theorem to the quantum settings, or to disprove it, have encountered 
severe obstacles so far \cite{Aha2, Arad, Has, Has2}. 
In particular, the classical proofs rely heavily on copying classical 
information; we cannot do the same in the quantum world, due to the 
no-cloning theorem. 
It seems that we do not understand some crucial issue about 
the local versus global behaviors of entanglement which is central 
to the problem. 
The resolution of the quantum PCP conjecture, either to the positive or 
the negative, is likely to shed light on
multiparticle entanglement, with implications  
to computational complexity as well as to physics, e.g., to 
hardness of approximation in the quantum setting,  
robustness of quantum correlations at room temperature, 
quantum error correction, topological order 
and more. 

Given the difficulty in making progress on the quantum PCP question so far, 
it is natural to restrict attention to commuting Hamiltonians.  
One one hand, one might hope to prove a quantum PCP theorem for 
commuting Hamiltonians, with the hope of borrowing ideas from such 
a proof to the more general case (this route might be useful  
even if eventually the general CLH problem 
turns out to be inside NP, which would collapse the commuting PCP theorem 
into the already known classical PCP theorem). 
On the other hand, upper bounds on the hardness of approximating
the CLH problem (such as putting the CLH approximation
problem inside $NP$ for certain approximation factors) 
could be considered as initial steps towards disproving 
the general PCP conjecture, or at least as limitations to its strength. 
This approach of attacking first the commuting case 
was taken by Aharonov and Eldar \cite{Aha}, Hastings \cite{Has}, 
and Hastings and Freedman \cite{Has2}.  

{~}

\noindent{\bf Locally Testable codes} 
Classical PCP proofs rely heavily on a particular type of error correcting 
codes, called LTC codes (Locally Testable Codes) which are interesting 
in their own right.  
These are codes which satisfy the following property: 
Given a word of fractional distance $\delta>0$ from the code, 
we would like to detect that it is not in the code by randomly 
choosing a constant number $k$ out of the $n$ locations, 
and performing a simple test on those locations.  
The code is called a locally testable code (LTC)
if such a test succeeds with 
probability $s>0$, given that $\delta>0$ (there are 
variants on this definition, which we will discuss below).
Note that the requirement for soundness when querying a constant number 
of locations 
is reminiscent of the similar goal in the context of PCP, 
in which querying a 
constant number of bits of a proof is supposed to detect that the proof 
is incorrect with constant probability.  
Indeed, locally testable codes, such as the Hadamard and the 
long code \cite{Hastad, Rubinfeld} have played a central 
role in the proofs of the PCP theorem (\cite{Arora,ALMSS,Din}).

One can make the connection between the classical PCP theorem 
and LTCs more rigorous.   
It is possible to define a strong version of PCP, called 
PCP of proximity (PCPP) \cite{BGHSV}, 
which is a structure hidden in most known 
PCP proofs; Ben-Sasson et. al. showed \cite{BGHSV} 
how given a PCPP, one  can take a good code 
and construct from it a good LTC code, which inherits its local testability 
parameters (soundness and number of queries) from those of the 
PCPP \cite{BGHSV}.

To the best of our knowledge, quantum LTCs and their robustness were 
not studied before, and neither was the notion of quantum PCPPs.

\subsection{Our contribution}
In this paper we attempt to gain new insights into the above topics. 
To do that, we consider an important facet, central to
all the above issues: the topology of the interaction graph 
underlying the set of constraints. 

Constraint satisfaction 
problems on a lattice, both in the quantum and classical case, 
can be easily approximated in polynomial 
time to within factors 
which are even sub-constant. This is done by throwing away 
constraints that disconnect the lattice to 
logarithmic sized boxes which can be solved separately.
On the other hand when the underlying graphs of the classical problem
are expanders, the constraint satisfaction problem can exhibit 
extreme robustness against even constant approximation factors, 
as in the PCP theorem (see Dinur's proof \cite{Din,AB}). 

Expansion of the interaction graph plays a crucial role 
in the context of error correcting codes as well; 
Classical expander codes were defined in 
\cite{Zemor,Alon,Spi}. 
In the quantum regime quantum codes on expanders have not been 
explored very thoroughly, though see \cite{Zemor2, Has2}.

In the above mentioned results, expansion is referred to loosely; 
several definitions of expansion are used in the literature, and we have 
not specified which one is being used. To be concrete, 
we need to choose a good expansion definition for our context. 
We first observe that crucially, 
the graphs we work with are $k$-local hypergraphs for $k\geq 3$.
This because the two-local case of commuting Hamiltonians is 
not interesting from our point of view; 
$2$-local CLH is known to be in NP \cite{Bra}, and exhibits only local 
entanglement, and stabilizer codes with $2$-local check terms cannot 
correct quantum errors for similar reasons.
We therefore need to work with hypergraphs, or with CLHs with $k\ge 3$ 
locality. 

One possible route to take is to use one of the known
definitions for expansion in hypergraphs: geometric, homological or spectral; 
(for references and a survey see \cite{Parzan}).
However it is unclear which of those 
makes more sense and the relations between them are not all known \cite{Parzan}.
One could also choose to work with a graph underlying the hypergraph 
(connecting any two nodes that appear in the same constraint); 
This definition loses much of the "structure" of the original hypergraph. 
We choose to work here with
the natural bi-partite graph induced by the hypergraph (as
in e.g., \cite{Spi, CRVW}). 
We define the induced bi-partite graph $G=(L,R;E)$ with
constraints on the left, and 
variables / qudits on the right (see definition \ref{def:bipgraph}), 
and a variable is connected to the constraints it appears in. 
We say that such a graph is an $\epsilon$ small-set bi-partite 
expander, if for any set $S$ of size at most $k$ particles, the number of local 
terms incident on these particles is at least $D_R |S| (1-\epsilon)$, where 
$D_R$ is the right degree of the graph.  
We note that 
$D_R |S|$ is the maximal possible number of constraints acting on those 
particles, so $\epsilon$ can be viewed a ``correction'' to this number.

We remark that our definition cares about the expansion of only 
constant-size sets,
whereas usually, one refers to {\it small-set expanders} as those graphs where all sets of size some
linear fraction of $|R|$ are required to expand \cite{Din3,Ragh}.
We will address this point also 
in section (\ref{sec:discussion}), following the exposition of our results.
 
We make use of a very simple property of very good small-set
expanders, namely that most of the neighbors of a set of size $k$ 
only touch this set at one point (see Fact \ref{fact:essence} 
in Subsection \ref{subsec:overview}).  
This turns out to be an extremely useful property which will bare 
strong consequences for both the computational complexity of CLHs 
on small-set expanders, 
as well as for stabilizers whose graphs are small-set expanders.

\subsubsection{Approximating CLHs on expanders}
We start by studying the complexity of  
the approximate version of the CLH problem;  
we ask whether we can bound from above the complexity of CLH 
if we are allowed a constant error in the minimal energy; 
Specifically, if we allow to throw away a certain constant fraction of the 
constraints; For that matter, we study what seem to be 
``hardest'' class of interaction graphs, 
namely small set expanders. 

We show that 
the approximation version of the $CLH(k,d)$ problem 
becomes "easy'' and falls into NP. 
(This is of course more surprising if one believes that the general 
$CLH$ problem is not in $NP$; if on the other hand one believes 
$CLH$ is in $NP$, this can be considered as a step towards this goal.)

Our first theorem states that:

\begin{theorem}\label{thm:approxBPinNP}
(the approximation of 
CLH on small set bi-partite expanders is in NP) 
Let $\gamma(\eps) = 2kd \eps$.
Let $H$ be an instance of $CLH(k,d)$ for constant $k,d$ whose bi-partite 
interaction graph has a right degree $D_R$,
and is $\epsilon$-small-set expanding, for $\epsilon<\frac{1}{2}$.
Let $\lambda_1 > \lambda_2 > \hdots > \lambda_N \geq 0$ 
be the eigenvalues of $H$.
Then for every eigenspace $\lambda_i$, there exists a state $\ket{\psi_i}$
 such
that $\left| \left\|H \ket{\psi_i} \right\| - \lambda_i \right| 
\leq \gamma(\eps) \left\|H \right\|$, and such that 
$\ket{\psi_i}$ can be generated by a constant depth quantum circuit.  
In particular, the $\gamma(\epsilon)$-approximation problem of $CLH(k,d)$ 
on such $\epsilon$ small-set bi-partite expander graphs is in NP. 
\end{theorem}

For the above theorem to be non-trivial, 
$\epsilon$ must be sufficiently small, so that $\gamma(\epsilon)=2kd\eps$ 
is less than $1$. 

One might ask whether Theorem \ref{thm:approxBPinNP} is trivially true; 
 we might have allowed to
remove enough terms so that the graph is disconnected into 
small components and the problem becomes solvable in $P$, 
just like what happens in the lattice case.  
We show that this is not the case:  

\begin{theorem}\label{thm:NPhard}
Given $\epsilon<1/2$, it is $NP$-hard to approximate $CLH(k,d)$ 
whose bi-partite interaction graph is
$\eps$ small-set expanding, to within a factor 
$\gamma(\eps)=2kd\eps$. 
\end{theorem}

Theorem \ref{thm:approxBPinNP} can be considered as a first step 
towards proving that the CLH problem lies inside NP; 

Alternatively, if one takes the point of view that CLH is not inside NP,
and hence that a quantum PCP theorem might in fact be proven via 
showing hardness of approximations of CLHs 
(as is the approach of \cite{Has,Has2}), then  
Theorem \ref{thm:approxBPinNP} puts limitations
on such a qPCP.  
While standard constructions of classical PCP 
use graphs with excellent small-set expansion (see the proof of 
Theorem \ref{thm:NPhard}) and one can make these graphs have arbitrarily good 
small-set expansion (albeit while increasing the degree),  
Theorem \ref{thm:approxBPinNP}
implies that the smaller the expansion error of the graphs 
underlying the CLHs are, 
the weaker the qPCP theorem using these CLHs becomes.  
The approximation error for which hardness is shown 
is bounded from {\it above} by some constant times the expansion error. 
This contradicts classical intution. 
We discuss this further in the Discussion,  
Subsection \ref{sec:discussion}. 

\subsubsection{Quantum stabilizer Locally-Testable Codes}
Next, we 
study the much related topic of stabilizer codes. 
In particular, we study how stabilizer codes with local stabilizers 
behave in the context of local testability. 

As mentioned above,  
classical locally testable codes satisfy the following property: 
Given a word which is of relative distance $\delta$ from 
the code, then if $\delta$ is not $0$ then this fact can be detected 
with some constant probability $s>0$ by querying 
a {\it constant} number $k$ (hence, {\it locally} testable) 
of coordinates in the code; 
One can consider strong LTC codes in which the probability
$s$ is proportional to $\delta$, even for sub constant $\delta$'s,
or weak $LTC$ 
codes in which the probability $s$ is at least a 
constant given that $\delta$ is at least a constant. 
The essential parameter of LTC codes is thus its robustness, captured  
roughly by how $s$ relates to $\delta$. 
 
More precisely, 
we say that a locally testable code with degree $D_R$ (namely, that 
each qudit is examined by $D_R$ constraints)
is $r(\delta)$-robust if an error on $\delta n$ 
locations violates a at least $r(\delta)D_R\delta n$ 
constraints, namely, a fraction $r(\delta)$ of the maximal number 
it can possibly violate (See Definitions \ref{def:weight}, \ref{def:robust} 
for a rigorous 
definition of the weight of an error and the robustness for various weights). 


There are two famous classical LTC codes used in the PCP 
proofs of \cite{ALMSS, Din, AB}: 
the long code and Hadamard codes. They 
exhibit extreme robustness, in the following sense: 
for any $\eps>0$, there exists $\delta>0$ such that
all error patterns of weight at most $\delta n$ (for code-length $n$)
are violated by a fraction at least $1-\eps$ of their 
incident ($3$-local) constraints).

We would like to study the robustness of quantum error correcting codes; 
To the best of our knowledge, the question of local testability of 
Quantum error correcting codes was not studied yet. 
The very notion of weight of an error is not very natural in the context 
of general quantum error correcting codes; we
leave the general definition to a later version of this paper. 
For stabilizer codes, however, the notion of weight 
of an error is very natural: Just count the number of non-identity 
Pauli's in the error (essentially, modulo the stabilizer group).   
For the rigorous definition see \ref{def:weight}.  
Robustness would thus correspond to how many stabilizers will not 
commute with such an error. 

The robustness of quantum stabilizer codes with local stabilizers seems, 
on the face of it, to be inherently restricted. 
It is illuminating to consider in this context the Toric code example
\cite{Kit2}. 
The Toric Code is defined by a set of $4$-local constraints
on a lattice of $\sqrt{n}$ by $\sqrt{n}$ qubits.  
It is well known that one can consider a chain of errors of 
length $\theta(\sqrt{n})$ of qubits, such that the only violated 
constraints will be those two constraints that touch the 
end qubits of this chain.
This already indicates very 
low robustness ($r(\delta)< \frac{1}{\sqrt{n}}$) for values of $\delta$ 
of the order of $\frac{1}{\sqrt{n}}$.    
It is easy to extend the low robustness behavior to words of linear distance 
from the code,
by scattering short chains of errors of length 
$\theta(1/\epsilon)$, which are separated from each other by at least
 $\theta(1/\epsilon)$ lattice sites. 

The message of this part of our paper is that quantum stabilizer codes 
in general, as long as they have local stabilizers, 
exhibit inherent non-trivial (and inherently quantum) 
upper bounds on their robustness. 

Our first result in this context might be viewed as 
surprising, as once again it poses 
constraints when the underlying interaction is expanding: 
it turns out that if the underlying graph 
is a good small set bi-partite expander, then the robustness 
is as low as the expansion error itself, 
for all errors with weight smaller than some constant $\delta$. 

\begin{theorem}\label{thm:QECC}  
Let $C$ be a stabilizer code, with
minimal distance $>1$, and a $k$-local 
generating set ${\cal G}\subset \Pi^n$,
such that each qudit is examined by $D_R$ generators.
Suppose the bi-partite interaction graph of ${\cal G}$ is $\eps$-small 
set bi-partite expander, 
for $\eps < 1/2$. 
Then,  for all $\delta < \min\{\frac{1}{k^3D_R},\frac{1}{2}dist(C)\}$, 
we have   
$r(\delta)\le 2\eps$.
\end{theorem}

In the above $dist(C)$ denotes the distance of the code
(see Definition \ref{def:stab}).  
Theorem \ref{thm:QECC} thus 
provides a trade-off between the expansion and the robustness; 
The better the expansion, the worse the robustness. 
This seems like an inherently quantum phenomenon, contrary to 
what happens in the classical world.  

To contrast it with the classical setting, recall 
a result by Dinur and Kaufmann
which shows that a robust LTC code must have 
an interaction graph which is a small-set 
expander (Theorem 1.1 in \cite{Din2});  
We note that direct comparison does not hold since the 
definitions of expansions used are different; 
\cite{Din2} does not use bi-partite graph expansion but rather the 
graph in which an edge connects any two nodes that participate in a common 
constraint. 

Stronger evidence that
Theorem \ref{thm:QECC} shows a solely quantum phenomenon 
is given by the construction of {\it lossless expanders} in \cite{CRVW}.
This construction gives
classical error correcting codes whose robustness is very close to 
$1$. 

\begin{claim}\label{cl:classical}
For any $\eps\in (0,1/2)$, and $r\in (0,1)$ there exists 
a constant $\delta=\delta(r,\epsilon)$, such that there exists
an explicit infinite 
family of codes $\left\{C_{\eps}(n)\right\}_{n\in \mathbf{N}}$, of $n$ bits,
of constant fractional rate $r$, and constant fractional distance 
$d=d(\epsilon,r)$, 
whose check terms are of locality whose {\it expectation}
is equal to a constant $k$, 
and all errors of weight less than $\delta n$ have robustness 
at least $1-3\eps$.
\end{claim}

The proof of this claim is given in the appendix. 
We note that contrasting Claim \ref{cl:classical} with Theorem 
\ref{thm:QECC} is not exact, since 
the classical codes constructed have only average constant degree and not 
constant degree; we believe this is not an essential point. 

We use theorem \ref{thm:QECC} together with 
a rather involved probabilistic argument, to establish  
an {\it absolute} non-trivial upper-bound on 
robustness of stabilizer codes, for errors of size at most 
some constant relative distance. A semi-trivial such bound exists,
due to the size of the alphabet:  
For a given qudit of dimension $d$, the number of possible errors 
on one qudit is $d^2-1$. Fix a qudit, and pick $Q$ to be 
the Pauli on that qudit which is most popular 
among all the stabilizers acting on that qudit.
Then at least $1/(d^2-1)$ of the stabilizers on that qudit will 
commute with that error, and so at most $\alpha(d)= 1- \frac{1}{d^2-1}$ 
 of the stabilizers touching that qudit will not commute with that error, 
and thus will detect the error (see fact \ref{fact:alphabet} 
for an exact proof). The robustness is thus bounded by 
$\alpha(d)$, which we call  
the {\it single-error robustness}, or the alphabetical 
upper bound on the robustness.  
Classically, there is no direct analogue to the requirement of non-commutativity to achieve constraint violation, and hence
 $\alpha(d)=1$ in the classical case. 

We are interested in a deeper phenomenon; it turns out that for any 
$k$-local 
stabilizer 
codes, the robustness must be strictly smaller than 
the {\it single-error robustness}. 

\begin{theorem}\label{thm:robust} (Roughly)
For any stabilizer code $C$
of $k$-local terms ($k\geq 4$) over $d$-dimensional qudits, where
each qudit interacts with $D_R$ local terms,
errors of weight at most some constant fraction of $n$
have robustness at most $\alpha(d)(1-\gamma_{gap})$ for some 
$\gamma_{gap}=\gamma_{gap}(k,d)>0$. 
\end{theorem} 

This indicates some inherent ``problem'' in the quantum setting 
when one is interested in maximizing the robustness to make the 
code most 'sensitive'' (in terms of energy penalty) to errors; 
Understanding this phenomenon might be tightly related to 
the clarification of the qPCP conjecture. 

\subsubsection{PCPs of Proximity}
A $PCP$ of proximity, or $PCPP$, is a form of language verification where the verifier receives, in addition to the
usual NP-witness to the veracity of the claim made by the prover, an auxiliary proof that allows for a low-query test on whether or not the witness is "close" to or "far" from the set of acceptable witnesses.
Ben Sasson et al \cite{BGHSV} provided  construction of 
an LTC code given a PCPP, and described how the parameters are mapped 
from one to another. (see Construction 4.3, and Proposition 4.4
in \cite{BGHSV}). 

We define quantum PCPPs, where the PCPPs must be 
defined with respect to a given set of constraints that can be used.
We then show that a similar result to that of \cite{BGHSV} 
mapping a quantum PCPP to an LTC with a corresponding 
set of constraints, holds also in the quantum setting.  
Thus, limitations on quantum LTC codes translate to limitations 
on quantum PCPPs using similar sets of constraints; 
providing an important motivation to understand Quantum LTC codes. 

We will delay the definition of quantum PCPP using 
various sets of constraints, together with the general definition 
of LTCs (which use not only stabilizer constraints) and the above statement 
connecting the two, 
 to the next version of this paper. 

We mention that  
for now, our bounds on robustness of LTC codes are merely initial indications
of the difference between the quantum and the classical LTC behavior, 
and much is left for improvement in terms of the parameters; moreover,
currently our bounds on the robustness of LTCs hold only for 
stabilizer codes. Thus, they imply rather 
limited implications on qPCPPs. 
However, the connection between qLTCs and qPCPPs 
extending the known classical connection between LTCs and PCPPs  
strongly motivates further 
study of qLTCs, as well as improving our parameters and extending our results 
to other constraints systems; a deeper study of the implications 
to the qPCP conjecture of this direction is called for.  

\subsection{Overview of Proofs}\label{subsec:overview}
\subsubsection{Simple observations on expanders}
Given a bi-partite graph $G(R,L:E)$, we say a set of nodes in $R$
is $\epsilon$-expanding if $\Gamma(S)\ge D_R|S|(1-\epsilon)$ 
where $\Gamma(S)$ the neighbors of $S$ and $D_R$ the right degree. 

Two basic observations in this paper are the following; 
The proofs can be found in the appendix. 
(see \cite{CRVW} for similar observations)

\begin{fact} \label{fact:essence}
Consider $S\subseteq R$  in a bi-partite graph $G(R,L:E)$ 
and let $S$ $\epsilon$-expanding, for $\eps<\frac{1}{2}$.
Then a fraction at most $2\eps$ of all
vertices of $\Gamma(S)$ have degree strictly larger than $1$ in $S$.
\end{fact} 

\begin{fact}\label{fact:deg}
Let $S\subseteq R$ in a bi-partite graph $G = (R,L;E)$, such that $S$ is $\eps$ expanding.
Then there exists a vertex $q\in S$, 
such that the fraction of neighbors of $q$ with at least two neighbors in $S$
is at most $2\eps$.
\end{fact}

\subsubsection{The complexity of approximating CLH on expanders}
The idea behind the proof
 of Theorem (\ref{thm:approxBPinNP}) is as follows.  
Much of our current knowledge on the complexity of CLH 
relies on a critical observation by \cite{Bra} that
any two commuting local terms $H_i,H_j$ can be viewed as acting on "disjoint" 
subsystems, using some basic facts from the representation theory of 
$C*$-algebra. 
Informally (see Subsection \ref{sec:bv}
 for a formal description), 
the idea is this. 
We consider two terms in the Hamiltonian, $H_i$ and $H_j$, which 
intersect on some subset of qudits, whose Hilbert space is 
${\cal H}_{int}$. Bravyi and Vyalyi's lemma claims that 
there exists a local isometry on ${\cal H}_{int}$, preserved by $H_i,H_j$, which if applied, the space of the intersection can 
be written as a direct sum of subspaces, and in each one, the two terms are 
"disconnected" from each other, so that
they essentially act on different subsystems which are in tensor product.
In other words, 
\begin{equation}\label{eq:sum}
{\cal H}_{int} = 
\bigoplus_{\alpha} {\cal H}_{int}^{\alpha,i}\otimes {\cal H}_{int}^{\alpha,j}
\end{equation}

Such that $H_i$ ($H_j$)
restricted to the $\alpha$ subspace acts non trivially 
only on the left (right) 
subsystem, ${\cal H}_{int}^{\alpha,i}$ (${\cal H}_{int}^{\alpha,j}$). 
This lemma was used by \cite{Bra} to show that 
$2$-local CLHs (corresponding to graphs) are in NP; 
However as was noted in \cite{Aha} more is needed 
in the case of hypergraphs. 

In the context of expanders, we 
use it as follows. 
Given a constraint $H_i$, 
if we now remove all terms that share at least $2$ vertices with 
$H_i$, (and this, by Claim \ref{fact:essence} is just a small portion of 
the terms intersecting that term) 
we get into a situation in which all terms that intersect $H_i$, 
intersect it at only one particle. This makes the intersection 
pattern simple enough, so that 
Bravyi and Vyalyi's lemma \cite{Bra} can be applied almost directly 
on each of the qudits on which the terms acts on.   
The final result is that the term $H_i$ can be separated 
from the remaining terms and can be viewed (up to the relevant 
isometries) as acting on a separated subsystem. 

How to use this insight to show that the approximation to within 
a small constant is in NP? 
It is not enough to simply take away, for each term, all terms 
that make it isolated, as this will require removing too many 
terms. 
Instead, we use the above idea iteratively.  
At each iteration we choose a local term $v$ and ``isolate'' it 
(see definition \ref{def:isolate}): 
we remove as many terms as necessary 
so that among the remaining terms, 
no local term shares more than one particle with $v$ (see figure \ref{fig:ex1}).
Bravyi and Vyalyi's lemma can then be applied, 
and $v$ can be separated from the rest, after restricting 
the qudits to the relevant subspaces in their direct sum of 
Equation \ref{eq:sum}. Now we can continue in the same way, starting
from the restricted subspace. 
This way we gradually "tear-away" local subspaces of particles. 

The analysis of the upper bound on the number of terms 
that need to be removed altogether requires some thought, since 
the number of particles does not necessarily decrease through 
one iteration; it is only their dimensionality that decreases. 
We resort to an amortized analysis 
counting the number of {\it dimensions} that are removed in total. 
As far as we know this amortized counting of dimensions is novel
in the context of quantum Hamiltonian complexity, and
may be of interest on its own. 

Note that unlike in the lattice case, 
removing edges here is done not 
in order to {\it disconnect} part of a graph, but 
only to isolate terms; loosely speaking, to make the interaction pattern
more sparse. Disconnecting entire parts of the graph will require 
removing too many terms in the context of expanders, and will 
thus lead to too large an approximation error. 
The proof of theorem (\ref{thm:NPhard}) that the approximation problem on small-set expanders is NP-hard follows from the observation that the output of the PCP amplification routine due to Dinur (\cite{Din}) is an excellent small-set expander.

\subsubsection{Bounds on LTC codes on Expanders}
To prove Theorem \ref{thm:QECC}, 
we argue that a stabilizer code whose underlying 
interaction graph is a good enough bi-partite small-set expander, 
cannot be too robust. To do this, we define a particular error pattern, 
which has a large weight modulo the stabilizer, but which does not 
violate too many generators.

The error pattern we define acts on one qubit in each one of a set of 
 "isolated" terms. By ``isolated'' term (formally, 
$L$-independent terms - see Definition \ref{def:Lind}), we mean 
something completely 
different than the previously used isolated qudits;
a set of terms is $L$-indepdent if no two terms in it are closer than 
some constant in the interaction graph. 
Picking one qudit in each term, we claim that the weight of this error 
cannot decrease when multiplied by an element from the group. 
Intuitively, this is 
because of the following reason; 
If such an error pattern can be represented more succinctly, there must 
be one single qudit whose error is removed 
modulo the stabilizer group, 
and moveover, no new errors appear (modulo the group, again) 
in a constant but large enough neighborhood of that qudit. 
We show that this cannot happen since the stabilizers are too local. 

We need to design an
error pattern whose robustness will be limited. 
By examining the code itself, we can "distill" an error which
causes a minimal number of violated constraints.
Here Fact \ref{fact:deg} allows to pick  
cleverly the qudit in each term on which the error will act, and also 
to define what Pauli error will act on it to minimize the penalty 
(namely, the number of generators which will not commute with the error).  
For a term $g$ we choose the qudit $q$ to be the one promised by the lemma: 
a qudit s.t. a vast majority of the terms acting on it intersects $g$ 
at only one location. We define the error on $q$ to be 
the restriction of $g$ to $q$; since all stabilizer terms commute with $g$, 
those that intersect $g$ only on $q$ must agree with it on $q$. 
It follows that most of the incident generators on $q$ 
will commute with this error, and 
only a meager fraction of its incident generators, will realize that this is
indeed an error.

\subsubsection{Upper bound on robustness} 

Our final theorem states that for any constants $k\ge 4, d$,
there exists a constant less than $1$, 
upper-bounding the robustness
of any quantum LTC of $k$-local constraints, on $d$-dimension particles.
Our proof relies on bounding the robustness of a quantum LTC from two sides.
On one hand, we use the bound of theorem (\ref{thm:QECC}) 
which is, in fact, the "surprising"
side, which implies that high expansion forces low robustness.
We then add a new claim showing that quantum Stabilizers, not only 
suffer from the
quantum effect of (\ref{thm:QECC}) but also, cannot avoid, the "classical" 
effect that codes with poor expansion have low robustness.

We restrict our attention to sets of qudits belonging to generators which are "far" away from each other in the interaction graph of the generating set.  
Given a set of qudit "islands" with poor collective expansion, we try to find an error on this set which has less-than-optimal number of violations.
At least intuitively, a set with poor expansion, cannot have full robustness, since if "all" bits of the set are erroneous there are less than an optimal number
of constraints examining this set, so there would be less-than-optimal number of violations.
However, this intuition on its own is insufficient, since we don't know whether or
not we can find an error whose weight will be the entire set.
Rather, we resort to using a random error, so that only a small fraction of the set is erred.
We show, as a sub-claim that on average, even a random error which is "sparse" on the set, will
still sense the fact that the complete set has sub-optimal expansion, and would yield less than optimal number of violations.

The fact that the resulting error has large weight modulo the centralizer follows from an observation we prove, that each "island" of qudits cannot have its errors "erased" via a more succinct representation, if most of the "island" is uninhabited, i.e. less than $1/2$ of the qudits are erred.
We then have a fine trade-off: on one hand, we would like the number of errors on a given "island" to be smaller than $1/2$, and on the other hand, we would like a good portion of "islands" to have at least $2$ errors, so that they "sense" the sub-optimal expansion.   

\subsection{Comparison with prior work}\label{subsec:compare}
The question of approximating the minimal energy of local Hamiltonians
was considered by various researchers; see Bansai, Bravyi, Terhal 
\cite{BansaiTerhal},
and Kempe and Gharibian (\cite{GK}), 
who gave upper bounds on non-trivial approximation factors for general 
local Hamiltonians. 

The approximation of commuting local Hamiltonians, and the PCP 
conjecture in general, was the motivation of 
Hasting's work \cite{Has} in which 
he showed that certain classes of CLHs lie in NP (without 
approximation). Arad \cite{Arad} also considered the commuting case as a base 
for perturbations in his partial no-go for quantum PCP result.
Hastings and Freedman \cite{Has2} recently defined
generalization of expanders to hypergraphs, motivated by the attempt to prove 
the quantum PCP conjecture. 
We note that the graphs designed in \cite{Has2} are small-set expanders 
to which our theorems, and in particular Theorem \ref{thm:approxBPinNP}
apply; hence, our work puts limitations on the strength of quantum 
PCPs using such graphs.

In a recent independent 
result, Brand{\~a}o and Harrow \cite{HB} have shown that 
the approximation of $2$-local Hamiltonians on expanding graphs 
(here expanding in the standard definition) 
lies in $NP$, with
an error depending on the second eigenvalue of the interaction graph.
This is a complementary result to our theorem 
\ref{thm:approxBPinNP}, since on one hand, there is no commutation 
requirement in \cite{HB}, 
but on the other hand they only handle the case of graphs, 
while higher locality might be crucial in the context of the 
problems studied here. In particular, it is unclear whether it is possible 
to derive a result about approximating 
$k$-local Hamiltonians on expanders from the results of \cite{HB}, 
since the gadgets allowing to move from $k-$ to $2$-local \cite{Terhalk} 
change the topology of the graph.  

We note that our Theorem (\ref{thm:approxBPinNP}) reduces in 
$2$-local case to a $0$ error approximation, 
namely an exact verification procedure; in other words, it 
reduces to the result of (\cite{Bra}). This is because 
bi-partite expansion is essentially 
maximal for any $2$-local $CLH$, since in this case two 
different constraints cannot intersect on more than on qudit. 
This is another indication that at least for the CLH case, 
bi-Partite expansion (Definition (\ref{def:expbi})) 
seems to be the appropriate one to 
use, compared to that of regular expansion, which is used by \cite{HB}
for the general LH case. 

As for quantum error-correcting codes, quantum LDPC (Low Density Parity Check) 
codes, namely codes with local check terms, were already 
studied in the literature (see \cite{Zemor2,Zemor3} and references therein, 
as well as \cite{Kovalev1,Kovalev2}). 
However, as mentioned before, we are not aware of any prior work on 
quantum locally-testable codes.

\subsection{Discussion}\label{sec:discussion}

\paragraph{Complexity of CLH}
One might hope that Theorem \ref{thm:approxBPinNP}
could be extended to 
show that the approximation of the CLH problem on a {\it general} topology
is in NP, perhaps by bridging between our results on expanders to the 
easy case of approximation on lattices. 
We note that a hint about the difficulty of such a result comes 
from the status of the Unique-Game-Conjecture \cite{Khot} which is 
known to be easy on both cases (albeit for interaction graphs 
and not hypergraphs), though conjectured to be 
$NP$-hard in general. 
Still, it may be possible to achieve some weaker 
approximation of general graph using ``bridging'' between the case 
of lattices and expanders, perhaps by
sub-exponential witnesses, following the work of \cite{Tre},\cite{ABS}.
We note here that the "bridging" that we did achieve for stabilizer code in Theorem (\ref{thm:robust}), indeed
depends on these instances being more structured, 
and as such, does not follow
through for general CLH instances.
It thus remains an important open problem to generalize our results to 
all interaction graphs. 
Of course, clarifying the exact case remains a major open problem. 

\paragraph{Quantum PCP and quantum LTCs}
Though one cannot derive using Theorem (\ref{thm:approxBPinNP}) a bound for general PCP,
it does place severe restrictions on any commutative version of a quantum PCP: i.e. a poly-time
algorithm taking any instance of $LH$ (or even just $CLH$), 
to a $CLH$ instance with constant promise gap.
Specifically, assuming $CLH$ is not contained in $NP$, 
our results indicate that the 
output of such a quantum PCP is unlikely to assume the form of 
a bi-partite small-set expander, unlike in the classical PCP proofs. 
First, Theorem \ref{thm:approxBPinNP} implies that 
for a given expansion error $\eps$,
the promise gap must be of the order of $\eps$ which can be a 
very small constant.
Moreover, one may speculate that if the output of such a PCP routine, has some small expansion error,
than this error may be set arbitrarily low.  
This would imply that as $\eps$ goes to $0$,
 the promise gap goes to $0$, contradicting the existence of such $PCP$s.

We note that to the best of our knowledge, the question of general 
upper bounds on robustness 
of quantum LTCs was not considered before, though of course 
robustness of specific quantum codes is very well understood; 
interestingly, the prominent example
of the Toric codes has inherently small robustness, which is tightly related 
to the fact that the code exhibits topological order, and thus only the
``end points'' of an error matter \cite{Kit}.   
Quantum error correcting codes can be viewed as a generalized topological 
order (see \cite{Has}), and perhaps some bounds on robustness 
phenomena can be carried over 
to all codes, not necessarily stabilizer LTCs
as in our Theorem \ref{thm:robust},
perhaps by using the fact that the code has large distance;  
It will be very interesting to derive such a result. 
Also, it will be most interesting to come up with explicit
(or even non-explicit) constructions of quantum 
codes with very good robustness. 

Of course, strengthening of our results in various other directions 
are called for: be it generalizing to all weights of errors, 
to various types of constraints, and of course, removing the commuting 
requirement. 

In a later version we will  define Quantum PCPP; 
and make explicit the connection between quantum PCPP and quantum LTCs.  In particular,
this implies that the restrictions on LTC codes with various types of constraints lead to bounds on PCPPs 
using similar constraints. 
The results we have presented here (\ref{thm:QECC}) and (\ref{thm:robust}) provide unexpected limitation on
stabilizer LTC codes.
The connection to quantum PCPP strongly motivates extensions and further 
clarifications of our results on qLTCs and their implications to qPCPs. 

Finally, we believe the results we presented here
hint that possibly, the problem of CLH and clarification of 
the importance of the commutation relations, might be more relevant to the 
resolution of the quantum PCP conjecture than was thought before. 
For example, if the qPCP conjecture were true, and the method 
of proof was indeed through a Quantum PCPP, and if indeed CLH is not 
Quantum NP hard (all of the above conjectures might seem quite reasonable
to believe in at this point),   
then this must imply the existence of novel excellent qLTC 
codes whose local check terms are non-commuting, by which 
such a qPCP theorem would be proven. Can such codes exist? 
Studying this question directly can shed light on the likelihood of 
the above mentioned conjectures to be true. 

{~}

\noindent
\textbf{Organization of paper} 

In Section \ref{sec:bg} we provide the background. 
Section \ref{sec:appBP}  
proves theorem \ref{thm:approxBPinNP}, showing that approximation 
of CLH on expanders lies in NP.  
Section \ref{sec:NPhard} proves Theorem \ref{thm:NPhard}, 
showing that the approximation problem we handle is NP-hard. 
Section \ref{sec:QECC} provides bounds on the robustness of 
quantum LTC codes on expanders, and 
Section \ref{sec:robust} provides an absolute bound on robustness 
of stabilizer LTC codes regardless of the expansion of their 
underlying graph. 

\section{Background}\label{sec:bg}

A $(k,d)$-local Hamiltonian on $n$ qudits
is a Hermitian matrix $H = \sum_i H_i$ operating on a Hilbert space 
${\cal H} = {\mathbf{C}^d}^{\otimes n}$ 
of $n$ $d$-dimensional particles, where we assume for all $i$, 
$\left\|H_i\right\|=1$ and  
$H_i = h_i \otimes I_{n-k}$ is Hermitian operating non-trivially only on 
$k$ particles. 

\begin{definition}
\textbf{The commuting local Hamiltonian problem}
In the $(k,d)$-local Hamiltonian problem on $n$ $d$-dimensional qudits 
we are given a $(k,d)$-local Hamiltonian 
and two 
constants $a,b$, $a-b = \Omega \left(\frac{1}{poly(n)}\right)$. We are asked to 
decide whether the lowest eigenvalue of $H$ is at least $b$ or at most $a$, 
and we are promised that one of the two occurs. 
In the commuting case, we are guaranteed that $[H_i,H_j]=0$ 
for all $i,j$, and W.L.O.G, we also assume that $H_i$'s are projections. 
$a$ and $b$ can be taken to be $0$ and $1$ in this case, respectively. 
\end{definition}

\subsection{Interaction graphs and their expansion} 
One can consider various definitions for the interactions graphs 
underlying a local Hamiltonian.

\begin{definition}\label{def:bipgraph}
\textbf{Bi-Partite Interaction Graph}
For a $(k,d)$-local Hamiltonian $H = \left\{H_i\right\}_i$ 
we define the bi-partite interaction graph $G=(L,R;E)$ as follows: 
$R$, the nodes on the right, correspond to the $n$
particles of ${\cal H}={\mathbf{C}^d}^{\otimes n}$ 
and 
$L$ corresponds to the set of local terms $\left\{H_i \right\}_i$. 
An edge exists between a constraint $l\in L $ and a particle $r\in R$ 
if $H_l$ acts non-trivially (namely, not as the identity) on 
$r$. Note that the left degree is equal to $k$; 
We denote the right degree to be $D_R$. 
\end{definition}

\begin{definition}\label{def:expbi}
\textbf{Bi-partite small-set expansion}
\noindent
A bi-partite graph $G=(L,R;E)$
is said to be $\epsilon$-small-set-expanding, if for every subset of particles 
$S\subseteq R$ of size $|S|\leq k$, $|\Gamma(S)| \geq |S| D_R (1-\epsilon)$.
\end{definition}

Bi-partite expanders have been defined and used e.g., in 
\cite{Spi}, \cite{CRVW} to 
construct locally-testable classical codes. Here 
we require expansion to hold only for small sets. 

\subsection{Commuting terms and \texorpdfstring{$C^*$}{TEXT}-algebras}\label{sec:bv}
We now state the lemma of Bravyi and Vyalyi \cite{Bra} precisely: 
\begin{lemma}\label{lem:BV}
Let $H_i,H_j$ be two local terms on Hilbert space ${\cal H}$, $[H_i,H_j]=0$.
Let ${\cal H}_{int}$ denote the intersection of $supp(H_i)\cap supp(H_j)$, where $supp(H)$ is the subset of qudits examined non-trivially by $H$.
Then, there exists a direct-sum decomposition
$${\cal H}_{int} = \bigoplus_{\alpha} {\cal H}_{int}^{\alpha},$$
where for each $\alpha$ we have:
$${\cal H}_{int}^{\alpha} = {\cal H}_{int}^{\alpha,i}\otimes {\cal H}_{int}^{\alpha,j},$$
such that
both $H_i,H_j$ preserve all subspaces ${\cal H}_{int}^{\alpha}$, and moreover,
for each $\alpha$ $H_i|_{\alpha}$ is non-trivial only on the Hilbert space ${\cal H}_{int}^{\alpha,i}$, 
whereas $H_j|_{\alpha}$ is non-trivial only on the Hilbert space ${\cal H}_{int}^{\alpha,j}$.
\end{lemma}

\begin{figure}[ht]
\center{
 \epsfxsize=6in
 \epsfbox{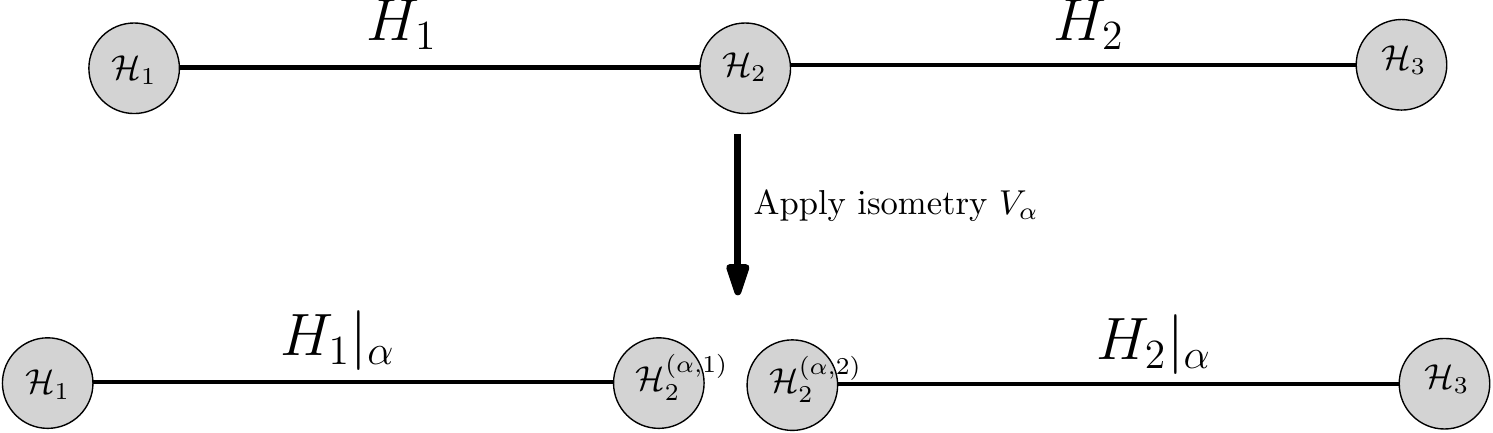}}
 \caption{\label{fig:ex2} 
An example of lemma (\ref{lem:BV}): a pair of $2$-local Hamiltonians, are restricted to a subspace $\alpha$ on their intersection.  Inside
this subspace, they act on separate subsystems.} 
\end{figure}

In this paper we apply Lemma (\ref{lem:BV}) 
in the following context: 
\begin{corollary}\label{cor:BV}
Let $H$ be a $CLH(k,d)$ instance, and let $H_0$ be a local term such that for any $H_i\in H \backslash \left\{H_0\right\}$ that intersects
$H_0$, the intersection is at most on one qudit.
Then there exists a direct-sum decomposition, on each of the qudits of $H_0$, preserved by all local terms of $H$, such that the restriction of $H_0$, to any such tensor-product subspace, is disjoint from any other term in $H$.
\end{corollary}

\begin{proof}
For each particle $q$ examined by $H_0$, 
let the term $H_1$ be the sum of all local terms acting on $q$ 
except $H_0$, and apply the lemma one by one on each qudit this way. 
\end{proof}

Note, that the above corollary does not require that 
any two local terms acting on a qudit in $H_0$, 
intersect only on that qudit. 

\subsection{stabilizer quantum error correcting codes}\label{sec:QECCdefs} 

\begin{definition}\label{def:stab} 
\textbf{stabilizer Code}
The group $\Pi^n$ is the $n$-fold tensor product of Pauli operators $A_1\otimes A_2 \otimes \hdots \otimes A_n$, where $A_i\in \left\{I,X,Y,Z\right\}$.
along with multiplicative factors $\pm 1, \pm i$ with matrix multiplication as group operation.
A stabilizer code $C$ 
is defined by a set of commuting elements in $\Pi^n$; 
the set is denoted by 
${\cal G}$.   
The group generated by ${\cal G}$ is an Abelian group denoted $A\subset \Pi^n$. 
The codespace is defined as the mutual $1$-eigenspace of all elements in $A$
(alternatively, in ${\cal G}$).   
An element of ${\cal E} \in \Pi^n$ is said to be an error if 
it does not commute with at least one element of ${\cal G}$. 
i.e. ${\cal E} \notin \mathbf{Z}(A)$, where 
$\mathbf{Z}(A)$ is the centralizer of $A$.
An element of ${\cal E} \in \Pi^n$ is said to be a logical operation, if 
it commutes with all of ${\cal G}$, but not generated by ${\cal G}$, i.e., 
${\cal E} \in \mathbf{Z}(A)-A.$
A stabilizer code is said to be $k$-local if each term $g\in {\cal G}$
is a tensor product of $n$ Pauli operators in which exactly 
$k$ are different than $I$.
\end{definition}

The Pauli group can be generalized to particles of any dimensionality, 
and thus the above definition can be generalized to work over any dimension 
$d$: 

\begin{definition}\label{def:generalpauli}
\textbf{The Pauli group generalized to $Z_d$}
Let $X^k_d:|i\rangle \mapsto |i+k~modulo~ d\rangle , 
P_d^{\ell}|j\rangle \mapsto w_d^{j\ell}|j\rangle$ 
be the generalized bit and phase flip operators on the 
$d$-dimensional Hilbert space, where $w_d$ is the primitive $d$th root of unity.
Let $\Pi_d$ be the group generated by those operators (and the necessary 
multiplicative factors, namely, all roots of unity of order $d$.) 
The group $\Pi_d^n$ is the $n$-fold tensor product of Pauli operators $A_1\otimes A_2 \otimes \hdots \otimes A_n$, where $A_i\in \left\{X_d^kP_d^\ell\right\}$
along with the multiplicative factors.  
\end{definition} 

The definitions \ref{def:stab} can be generalized 
to any dimensionality of particles $d$. 
We can now define the distance of the code: 

\begin{definition}\label{def:distance}
\textbf{Distance}
Let $C$ be a $k$-local stabilizer code with generating 
set ${\cal G}\subset \Pi^n$.
For an error ${\cal E} \in \Pi^n$, let $wt({\cal E})$ 
denote the size of the non-trivial support of ${\cal E}$,
and let $wt({\cal E}_{{\cal G}})$ denote its minimal 
weight {\it modulo} the group $A$ generated by ${\cal G}$, 
$C$ is said to have a constant relative distance $\delta>0$, if for any
${\cal E} \in \mathbf{Z}(A) - A$, we have 
$wt({\cal E}_{\cal G})\geq \delta n$.
\end{definition}

We note that a code with distance $1$ cannot correct any error 
(see \cite{Got})
and so we can assume the absolute distance is strictly larger than $1$. 
We also note that we can assume that there is no qudit $q$ and a state 
on it $|\alpha\rangle $such that 
all states in the code 
look like $|\alpha\rangle$ tensor with some state on the remaining qudits; 
We say that in this case, the qudit is {\it trivial} for the code.

\subsection{Notation}
Throughout the paper we shall use the following notation:
$d$ is the dimension of the qudits involved.
For a bi-partite graph we denote $G=(L,R;E)$ where $L$ denotes the 
left set of vertices of size $|L|=m$ (corresponding to constraints in the
 text), 
$R$ denotes the right vertices $|R| = n$ (corresponding to particles), 
and $E$ is the set of edges between $L$ and $R$.  
$D$ will denote the degree of a graph in more generality; 
$D_R$ will denote the right degree of a Bi-partite graph, 
which is assumed in this paper to be constant.

Given $S\subseteq R (L)$ in a bi-partite graph,   
$\Gamma(S)$ denotes the neighbor set of $S$ in $L (R)$. 
${\cal N}(q)$ will denote the neighborhood of $q$ in $R$, namely 
all the qudits participating in all the constraints acting on $q$. 
$\epsilon$ (and sometimes $\mu$) 
will be used to denote the expansion error for bi-partite 
graphs (as in Definition \ref{def:expbi}). 
$\delta$ will be used to denote the relative distance of a code. 
$\gamma$ will be used to denote the promise gap, or alternatively 
the approximation error of a $CLH$ instance, i.e. the fraction of terms we are allowed to throw away. 

\section{Approximate CLH}

\subsection{Approximating CLH on Expanders}\label{sec:appBP}

We can now prove Theorem (\ref{thm:approxBPinNP}). 
As mentioned in the introduction, 
the proof relies on the simple important observation regarding 
highly expanding small set bi-partite expanders, namely, 
that not too many terms intersect a term in more than one node 
(Claim \ref{fact:essence}). We start by proving the following: 

\begin{proof}(Of Claim \ref{fact:essence}) 
By definition, and by the requirement on expansion, 
the average degree of a vertex in $\Gamma(S)$ w.r.t.
$S$ is $\frac{D_R |S|}{|\Gamma(S)|}\le \frac{1}{1-\epsilon} \leq 1+2\epsilon$,
where the second inequality follows from $\epsilon< \frac{1}{2}$.
Since each vertex in $\Gamma(S)$ has degree w.r.t. $S$ at least $1$,
the claim follows. 
\end{proof} 

We now define a notion of isolation, which will help us handle 
single constraints: 

\begin{definition}\label{def:isolate}
\textbf{Isolated constraints and particles}
Let$(L,R;E)$ be a bi-partite  graph.
A constraint $g\in L$ is
\textbf{isolated}
if for any constraint $v\in L$ except $g$, 
we have $|\Gamma(v) \cap \Gamma{g}| \leq 1$.
We define the \textbf{isolation penalty} of
$g$, to be
the minimal number of constraints in $L$ that we need to remove 
so that $g$ is isolated. 
\end{definition}

We are now ready to prove the theorem:
\begin{proof}(Of Theorem \ref{thm:approxBPinNP})
Let $H$ be an instance of $CLH(k,d)$, and let $(L,R;E)$ be its bi-partite interaction graph, where the degree of each
$v\in R$ is $D_R$.
We perform an iterative process. 
At step $t$, we have a set of remaining terms $L_{rem}(t)$ 
acting on the remaining Hilbert space ${\cal H}_{rem}(t)$. 
We also have a set $L_{bad}(t)$ which are terms we have collected 
so far that we want to throw away. 
W.l.o.g., we shall find a state whose energy approximates the ground energy of $H$, but this can be applied to the approximation of any eigenspace.
Repeat the following: 
\begin{enumerate} 
\item \textbf{isolate}
Pick an arbitrary local term $v$ in $L_{rem}(t)$.   
and isolate it (Definition \ref{def:isolate})
by removing as few as possible terms from $L$
(putting them in a set of terms which we call $L_{bad}(t)$). 
\item \textbf{isometries}
Now that the term is isolated, the conditions of Corollary \ref{cor:BV} hold:
By the corollary, there is a tensor product of $1$-local isometries on the qudits of $v$: $W_v = \bigotimes_i W_i$ 
that is contained in the zero eigenspace of $L_{rem}(t)$,
such that after applying these isometries, the term $v$ acts on a disjoint subsystem
than the rest of $L_{rem}(t)$.
We conjugate $v$, and each term of $L_{rem}(t)$ by $W_v$.
We then remove any 
triviality we encounter, be it removing qudits from Hamiltonian terms  
or entire terms altogether. We call this {\bf Pruning}.
\item \textbf{update}
We set $L_{rem}(t+1)$ as the set of remaining 
local terms of $L_{rem}(t)$ after this pruning, 
and then update ${\cal H}_{rem}(t+1)$ as the support of $L_{rem}(t+1)$.
We set $L_{good}(t)$ to contain $v(t)$ restricted to the subspaces
of the qudits,  as in Corollary \ref{cor:BV}.  
\end{enumerate} 

We terminate when there is no longer any intersection between two different 
terms of $L_{rem}(t)$.
Clearly, this ends after at most polynomially many 
iterations, since the number of terms in $L_{rem}(t)$ decreases by at least $1$ 
each iteration. Let the number of iterations be $T$. 
We claim that a state {\it approximating}
the ground energy of the original Hamiltonian to within $|L_{bad}(T)|$,  
can be recovered from this 
procedure by finding the ground state of all terms in $L_{good}(T)$, 
and applying the inverse of the isometries applied along the way, which amounts to applying
a constant-depth quantum circuit.
This is true as long as the subspaces in the direct sum that were chosen 
contained the ground state of the current Hamiltonian; 
the $NP$ prover can provide the indices of the subspaces so that 
this holds. 


\begin{figure}[ht]
\center{
 \epsfxsize=6in
 \epsfbox{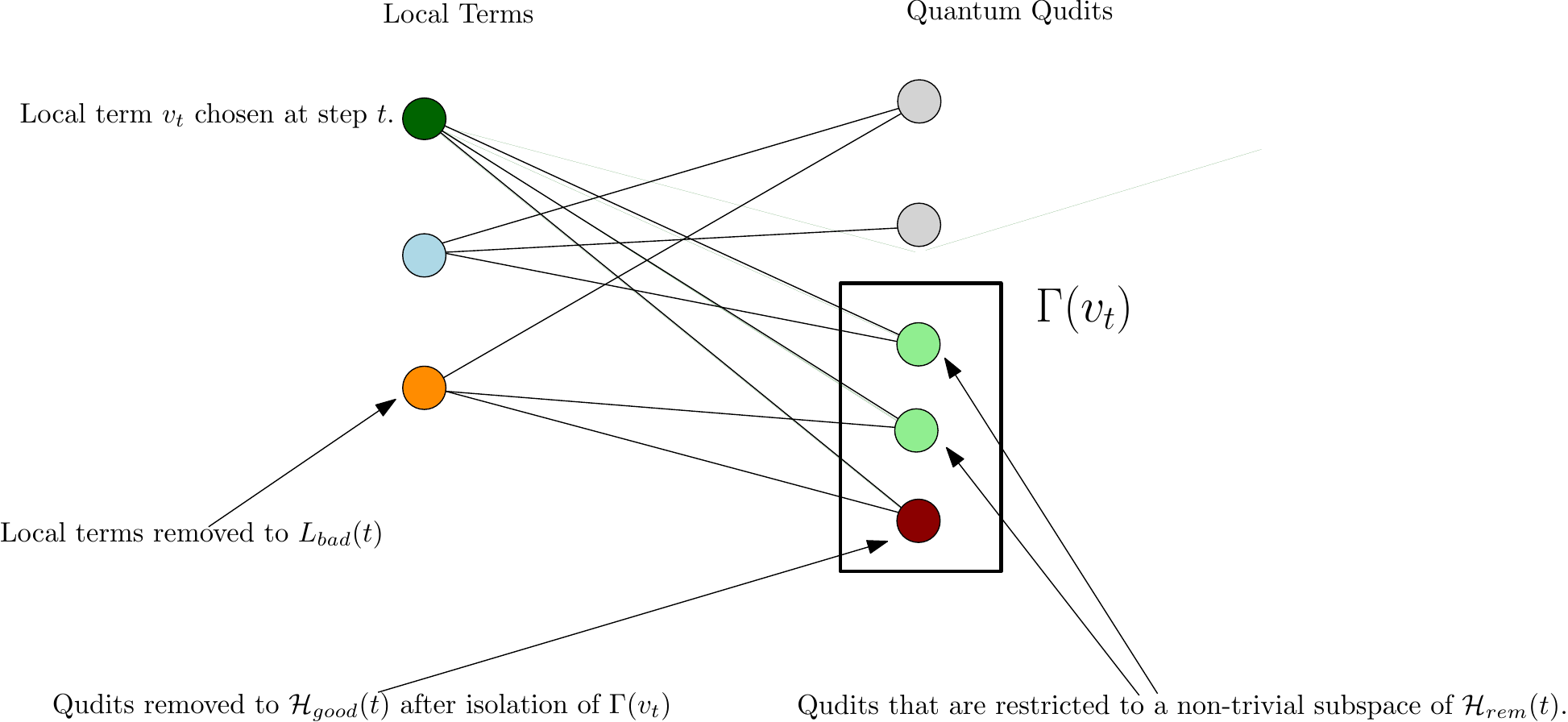}}
 \caption{\label{fig:ex1} 
The isolation procedure: after removing the local terms to $L_{bad}$, $v_t$ becomes isolated. 
We then apply local isometries to each of its qudits in $\Gamma(v_t)$.  One of the qudits is exclusively examined by $v_t$ after isolation,
so it is removed altogether to ${\cal H}_{good}$.} 
\end{figure}

\paragraph{Bounding the approximation error}
The non-trivial part of the proof 
is to upper bound $|L_{bad}(T)|$.
Let us measure the following ratio at each iteration $t$:
the number of local terms moved to $L_{bad}(t)$ during step "Isolate",
divided by the overall particle dimension (namely, sum of local dimensions
of particles) which we reduced.
We claim that this ratio is at most $2 D_R \epsilon$, and so
\begin{claim}
We remove at most $2 D_R \epsilon$ terms per local dimension we remove.
\end{claim}

\begin{proof}
To see this, consider a term $v(t)$ in $L_{rem}(t)$ whose support 
is a set $S$ with $a$ vertices. 
Let $G(t)$ denote the bi-partite interaction graph of $L_{rem}(t)$.
We know that in $G(0)$, it was true that  
$|\Gamma(S)|\ge aD_R(1-\epsilon)$, and so 
in order to isolate $S$, it sufficed to remove
 $2\epsilon D_R a$ 
vertices. 
Since $G(t)$ is derived from $G(0)$ by removing vertices and edges, 
this bound still holds. 
So at each step, when isolating a term of size $a$, we remove at 
most  $2\epsilon D_R a$ constraints. 

On the other hand, let us calculate how many qudit dimensions we remove 
at such an iteration, out of the total number of qudit dimensions (namely, 
the sum of dimensions over the particles). 
For each qudit $q\in S$, $q$ is either
removed altogether, and thus its local dimension 
is reduced by $dim({\cal H}_q)$,
or it is restricted to some non-trivial subspace 
and maybe further divided to two sub-particles, one of which is removed
from ${\cal H}_{rem}(t)$; 
this act reduces the local
dimension by at least $1$. 
Thus, the total number of local dimensions removed by one application
of "isolate'' is at least $|S|=a$.
\end{proof}

The initial total sum of particle dimensions is $d|R|$.
By the claim, $|L_{bad}(T)|\leq 2 D_R \epsilon \cdot (d|R|)$.
Substituting $D_R = k |L|/|R|$, we have that
$|L_{bad}(T)| \leq 2k \epsilon d |L|$, or $\frac{|L_{bad}(T)|}{|L|} \leq 2kd \eps$, as desired.
\end{proof} 

\subsection{NP-hardness of approximation}
\label{sec:NPhard}

The proof of theorem \ref{thm:NPhard} 
relies on the fact that the output 
of the PCP reductions in Dinur's proof \cite{Din} 
are excellent small set Bi-partite expanders. 
The gap can be amplified to, say, a $1/3$,
and yet the expansion error $\epsilon$ can be made smaller than 
any constant, and certainly smaller than the required $\gamma(\epsilon)$.   
This is an easy observation given the proof 
in \cite{Din} and we leave the details of this for the journal version.

\section{Bound on the Robustness of stabilizer LTCs 
on Expanders}\label{sec:QECC}
In this section we define stabilizer LTC codes and their robustness, and 
prove Theorem \ref{thm:QECC}. 

\subsection{Definitions: stabilizer LTC codes}
\begin{definition}\label{def:weight}
\textbf{The weight of an error on stabilizer codes}
Let $C$ be a stabilizer code on $n$ qudits, with 
generating set ${\cal G} \subset \Pi^n$. 
Let $A$ be the associated stabilizer group. 
Let ${\cal E}\in \Pi^n$ be some element in the Pauli group. 
We denote by $wt({\cal E})$ the number of locations in which 
${\cal E}$ is non-identity. 
We denote by $wt_{{\cal G}}({\cal E})$ the weight of ${\cal E}$ modulo the group 
generated by ${\cal G}$, i.e., the minimum over the weight of all words 
which are equal to ${\cal E}$ modulo an element in the group $A$ generated by 
${\cal G}$; 
Likewise, we denote by $wt_{{\cal Z(G)}}({\cal E})$ the weight 
of ${\cal E}$ modulo $Z(A)$, the centralizer of $A$, 
the group generated by ${\cal G}$. 
\end{definition}

We now define the robustness of a stabilizer code. 

\begin{definition}\label{def:robust}
\textbf{Robust Local stabilizer Codes}
Let $C$ be a stabilizer code, with 
generating set ${\cal G} \subset \Pi^n$, such that each particle is 
examined by $D_R$ generators.
$C$ is said to be $r(\delta)$-robust, 
if for any ${\cal E}\in \Pi^n$ 
such that $wt_{Z({\cal G})}({\cal E})=\delta n$,
we have that ${\cal E}$ does not commute with at least 
$r(\delta) \cdot D_R \cdot wt_{{Z(\cal G})}({\cal E})$ 
generators $g\in {\cal G}$ (we say that those generators, namely constraints, 
are {\it violated}). We call $r(\delta)$ the robustness of the code for 
errors of size $\delta$.  
We sometimes refer to the robustness of a given error pattern
${\cal E}$ as
$r$ if the number of violated constraints for that error pattern 
is at least $r \cdot D_R \cdot wt_{Z({\cal G})}({\cal E})$. 
\end{definition}

\subsection{A useful fact about stabilizers} 

\begin{definition}
\textbf{Restriction of stabilizers}
For a Pauli operator $A$, let $A|_q$ denote the $q$-th component of the tensor product $A$, and let $A|_{-q}$ denote the tensor product of all terms except the $q$-th.
Similarly, for a generating set ${\cal G}$, we denote by ${\cal G}|_q$ as the set $\left\{g|_q |g\in {\cal G} \right\}$, and similarly for ${\cal G}|_{-q}$.

\end{definition}

We now prove a useful fact: that the restrictions to a given qudit 
$q$ of all the generators of a stabilizer code with distance larger than 
$1$ cannot all commute.  

\begin{fact}\label{fact:nocommute}
Let $C$ be a stabilizer code 
with minimal distance at least $2$.
Then for any qudit $q$, there exist two 
generators $g(q),h(q)$ 
acting non-trivially on $q$ such that $[g|_q, g'|_q] \neq 0$.
\end{fact}

\begin{proof}
Otherwise there is a qudit $q$ such that 
for any $a,b\in {\cal G}$ we have $[a|_q,b|_q]=0$. 
Let $Q=g|_q$ for some $g\in {\cal G}$ s.t. $Q$ is not the identity. 
We have that $Q'=Q\otimes I_{-q}$, namely the tensor product with identity 
on the other qubits, commutes with all $g\in {\cal G}$, and thus 
is in the centralizer of ${\cal G}$: $Q'\in {\cal Z(G)}$. 
However, $Q'$ cannot be insider the stabilizer, since otherwise 
$q$ is in 
a constant state (the $1$ eigenvector of $Q$) 
for all words in the code, and is thus trivial for the 
code (see Section \ref{sec:QECCdefs}). 
Hence, $Q'\in Z(A)-A$, where $A$ is the stabilizer group 
and $Z(A)$ its centralizer, and so the distance of the code is $1$. 
\end{proof}

\subsection{Proof of Theorem \ref{thm:QECC}}
In the proof we will make use of "sparse'' sets of constraints, 
defined as follows. 

\begin{definition}
\textbf{$L$-independent set of constraints}\label{def:Lind} 
\noindent The neighborhood of a constraint is all constraints which 
intersect it in at least one qudit. 
A set of constraints $U\subseteq L$ is said to 
be $L$-independent 
if for any two constraints in the set, their two neighborhoods do not 
intersect on any qudit. 

\end{definition}

\begin{proof}(Of theorem \ref{thm:QECC}) 

\paragraph{Generating the error}
We want to construct an error with large weight modulo the centralizer 
that will not violate too many constraints in ${\cal G}$.  
Let $C$ be a stabilizer code 
with a $k$-local generating set ${\cal G}$, such that the
bi-partite interaction graph of $S$ is a small-set bi-partite $\eps$-expander. 
Let $U$ be an $L$-independent set of constraints of size $\delta n$.
We note that since $\delta \leq \frac{1}{k^3 D_R}$ 
an $L$-independent set of this size must exist, by a simple greedy algorithm. 
For a given $u\in U$, and $i\in [k]$, let $\alpha_i(u)$ denote the number of 
check terms incident on qudit $i$
of $\Gamma(u)$ that have at least two representatives in $\Gamma(u)$.
Then for each $u$ we define $q(u)$ to be the qudit of minimal $\alpha_i(u)$ over $[k]$.
Let $T = \left\{q(u) | u\in U\right\}$.
Let us define an error pattern:
$$ {\cal E} = \bigotimes_{u\in U} u|_{q(u)}.$$

We first note that ${\cal E}$ is not inside the centralizer 
$C({\cal G})$, and therefore it is 
an error; This is true since by fact (\ref{fact:nocommute})
for each qudit $q$ in the support of ${\cal E}$, 
${\cal E}|_q$ does not commute with $h|_q$ for some $h\in {\cal G}$.
But since $T$ is induced by an $L$-independent set, $h$ 
does not touch any other qudit in the support of ${\cal E}$ except $q$, 
so this implies $[h,{\cal E}]=[h|_q,{\cal E}|_q]\neq 0$. 
We will now show that ${\cal E}$ has large weight modulo the group, 
but is penalized by a relatively small fraction of the check terms.

\paragraph{Weight Analysis}
By definition, we have that $wt({\cal E}) = |T|=|U|=\delta n$.
We need to show
\begin{equation}\label{eq:tweight}
wt_{C({\cal G})}({\cal E})= |T|
\end{equation}
Since  $\delta$ was chosen to be smaller than half the distance of the code 
$C$, $wt_{Z({\cal G})}({\cal E})= wt_{\cal G}({\cal E})$ 
and so it suffices to argue about the weight of ${\cal E}$ modulo 
${\cal G}$. 
 
Suppose on the negative that $wt_{\cal G}({\cal E})< |T|$.
Let ${\cal E'}$ be an error pattern which is equal to 
${\cal E}$ modulo ${\cal G}$ but has weight $<|T|$. 
In other words, let $\Delta \in S$ be a word in the group 
$S$ s.t. $\Delta {\cal E}={\cal E'}$. 
Since the weight of ${\cal E'}$ is strictly smaller than 
that of ${\cal E}$, there must be one qudit $q_0$ in $T$, s.t. 
on the neighborhood ${\cal N}(q_0)$ the weight of ${\cal E'}$ is strictly 
smaller than that of ${\cal E}$, which is $1$; 
namely, ${\cal E'}$ must be equal to the identity on all the neighborhood 
of $q_0$.
Here, we have used here the fact that the neighborhoods of different qudits 
in an $L-$independent set are non-intersecting (definition \ref{def:Lind}). 
This means that $\Delta$ must be equal to the inverse of ${\cal E}$ 
on this neighborhood. But this inverse is exactly the following: 
It is equal to ${\cal E}|_{q_0}^{-1}$ on $q_0$, and to the identity on all other 
qudits in the neighborhood; and moreover, we have by construction 
that ${\cal E}|_{q_0}^{-1}$ on $q_0$, (and therefore also 
that ${\cal E}|_{q_0}$) does not commute with $h|_{q_0}$, for some $h\in {\cal G}$.)
This means that $\Delta$ does not commute with $h\in {\cal G}$, in contradiction 
to the fact that $\Delta \in S$. 

\paragraph{Robustness Analysis}
We upper-bound the number of generators that do not commute with ${\cal E}$.
For each $u\in U$, we have that the number of generators $g\in {\cal G}$ 
that do not commute with ${\cal E}|_{q(u)}$, is at most
the number of generators that share at least two qudits with $u$.
By fact (\ref{fact:deg}) there exists a qudit such that the fraction of its check terms with at least two qudits in $supp(u)$ is at most $2\eps$.
Since we chose $q(u)$ to be the qudit that minimizes that fraction over all qudits of $u$, the absolute number of such check terms is at most $2\eps D_R$.
Hence the overall penalty on ${\cal E}$ is at most $2 \eps |T| D_R$.
By Equation \ref{eq:tweight} 
we get that the penalty is at most $2\eps D_R wt({\cal E}_{\cal G})$.
By definition \ref{def:robust}, this implies that $r(\delta)=r(wt({\cal E}_{\cal G})/n) \leq 2 \eps$.
\end{proof}

\begin{corollary} \label{cor:QECC}
Let $C$ be a stabilizer code of minimal distance $>1$ 
and a $k$-local succinct generating set, such that each qubit is examined 
by $D_R$ constraints.
If there exists 
an $L$-independent set $U$, s.t. $|U| = \delta n$ for some $\delta>0$, 
and $U$'s neighbor set of qudits has expansion error at most $\eps$ 
(i.e. $|\Gamma(\Gamma(U))|\geq |\Gamma(U)| D_R (1-\eps)$), 
then for any $\delta'\leq \delta$ we have that 
$r(\delta')\le 2\epsilon$.
\end{corollary}

\begin{proof}
The proof follows exactly the proof of Theorem \ref{thm:QECC} noticing 
that the only thing we need from the assumption on the expansion of the 
graph, is the existence of an $L$-independent set whose expansion errors 
is at most $\epsilon$.   
\end{proof}

\section{An upper-bound on robustness}\label{sec:robust}

We now show an absolute constant, upper-bounding the robustness of any 
quantum stabilizer code $C$ with local generators.  
We start with an easy alphabet based upper bound.

\subsection{Alphabet-based bound on robustness}

In attempting to understand robustness of a stabilizer code, one must first account for
limitations on robustness that seem almost trivial, and occur even when there is just a single error:
\begin{definition}\label{def:single} 
{\bf Single error robustness}
Let $f(d)=d^2$ be the number of distinct 
generalized Pauli's on a Hilbert space of dimension $d$.
Let $t(d) = 1/(f(d)-1)$; The single error robustness in dimension $d$ 
is defined to be $\alpha(d)=1-t(d)$. 
\end{definition}

The motivation for the above definition is as follows. 
For any qudit $q$, there always exists $Q$, a non-identity Pauli operator 
among the $d^2-1$ Pauli's on a $d$-dimensional qudit,  
such that a fraction at least $t(d)$ of the generators touching 
$q$ are equal to $Q$ when restricted to $q$. 
If we consider a one qudit error on $q$ to be equal to 
$Q$, then it would commute with $t(d)$ of the generators touching 
$q$; and can at most violate $\alpha(d)$ of the constraints touching it. 
Thus, one can expect that it is possible to construct an error of linear 
weight, whose robustness is bounded by the single error robustness, 
using qubits whose neighborhood sets of constraints are far from each other.
Indeed, we show:

\begin{fact}\label{fact:alphabet}
\textbf{Alphabet bound on robustness}

\noindent
For any stabilizer code $C$ on $n$ $d$-dimensional qudits, of distance at least $2$, and a $k$-local succinct generating set ${\cal G}$, whose right-degree is $D_R$,
we have $r(\delta) \leq \alpha(d)$, for any $\delta \leq 1/(k^3 D_R)$.
\end{fact}

\begin{proof}
Similar to theorem (\ref{thm:QECC}) there exists an $L$-independent set $U$ of size $\delta n$.
For each $u\in U$ we select some qubit $q\in \Gamma(u)$ and examine the stabilizer component on $q$ of all stabilizers on $q$.
Consider $P(q)$, the restriction of all generators incident on $q$ to the qudit $q$, and let $MAJ(q)$ denote the Pauli that appears a maximal number of times in $P(q)$.
We then set ${\cal E} = \bigotimes_{u\in U} MAJ(q)$.
We first realize that ${\cal E}$ is an error: we want to show that there exists a generator $g$ such that ${\cal E}$ and $g$ do not commute.
Otherwise, 
${\cal E}$ commutes with all generators.
This implies that for all $q$, $P(q)$ consists of a single Pauli, contradictory to fact (\ref{fact:nocommute}).
Similar to equation (\ref{eq:tweight}) of theorem (\ref{thm:QECC}) the weight of ${\cal E}$ is $|U| = \delta n$.
Furthermore, for each qudit $q$, the fraction of generators on $q$, whose restriction to $q$ does not commute with ${\cal E}|_q$ is at most $\alpha(d)$, since the number
of appearances of ${\cal E}|_q = MAJ(q)$ in $P(q)$ is at least $t = 1-\alpha$.  Hence the number of violated constraints is at most $\alpha(d) \cdot |U| \cdot D_R$.
\end{proof}
 
\subsection{Separation from alphabet-based robustness}

In this section we show that this alphabet-based 
bound in fact cannot be achieved, and the robustness is further bounded 
and attenuated by a constant factor below the single qudit robustness, due 
to what seems to be a strictly quantum phenomenon. 
We will use the topology of the underlying graph 
to achieve this separation, by treating 
differently expanding instances and non-expanding instances.
Before stating the main theorem of this section, we require a new definition, and a simple fact:
\begin{definition}
\textbf{$k$-independent set of constraints}\label{def:kind} 
\noindent 
Let $H$ be a $k$-local Hamiltonian instance, and let $u$ be a constraint of $H$.
The $0$-th neighborhood of $u$, denoted $\Gamma^{(0)}(u)=\Gamma(u)$
is the set of qudits examined by $u$.
We define recursively,  
the $t$-th neighborhood of $u$, $\Gamma^t(u)$ to be the set of qudits, 
belonging to all constraints which act on qudits in $\Gamma^{(t-1)}(u)$.
A set of constraints is said to 
be $k$-independent 
if for any $a,b\in U$ we have $\Gamma^{(k)}(u) \cap \Gamma^{(k)}(v) = \Phi$.
\end{definition}
The following fact can be easily derived by a greedy algorithm:
\begin{fact}\label{fact:indset}
Let $\eta(k,D_R) = k^{-(2k+1)} D_R^{-(2k-1)}$.
For any $k$-local Hamiltonian $H$ on $n$ qudits, such that each qudit is examined by $D_R$ local terms, there exists a $k$-independent set of size at least $\eta n$. 
\end{fact}
\begin{proof}
Pick a constraint, remove all constraints in its $\Gamma^{(2k)}$ neighborhood, and repeat.
We get that the fraction of constraints is at least $k^{-(2k)} \cdot D_R^{-(2k)}$, and by the fact that 
the number of constraints times $k$, is equal to $n D_R$, we get the desired result.
\end{proof}

\noindent 
{\bf Theorem (\ref{thm:robust})}
{\it Let $C$ be a stabilizer code on $n$ $d$-dimensional qudits, of minimal distance at least $k$, and a $k$-local ($k\geq 4$)
succinct generating set ${\cal G}\subset \Pi^n$, where the right 
degree of the interaction graph of ${\cal G}$ is $D_R$.
Then there exists a function $\gamma_{gap}=\gamma_{gap}(k)> min\left\{10^{-3},0.01/k \right\}$ 
such that for any
$\delta \leq \min\{dist(C)/2,\eta/10\}$, we have 
$r(\delta')\leq \alpha(d) \left(1-\gamma_{gap}\right)$.
where $\delta' \in (0.99 \delta, 1.01\delta)$.
}

The proof of the theorem will use, on one hand, corollary (\ref{cor:QECC}) which upper-bounds the robustness of expanding instances, and on the other hand a lemma on non-expanding instances,which essentially, tries to "mimic" the behavior of the classical setting, in which non-expanding topologies suffer from poor robustness.
We now state this lemma:
\begin{lemma}\label{lem:indexp}
Let $C$ be a stabilizer code on $n$ 
qudits of dimension $d$, with minimal distance at least $k$
and a $k$-local ($k\geq 4$) succinct generating set ${\cal G}$, 
where the right degree of the interaction graph of ${\cal G}$ is $D_R$.
Let $\gamma_{gap} = \gamma_{gap}(k)=
min\left\{10^{-3},0.01/k \right\}$. 
If there exists a $k$-independent set $U$ of size $|U| = \delta n$,
with $\delta<\min\{dist(C)/2\}$,
such that the bi-partite expansion error of $\Gamma(U)$ is at least $\eps = 0.32$, 
i.e. $|\Gamma(\Gamma(U))|= |\Gamma(U)| D_R (1-\eps')$ for some $\eps'\geq 0.32$ then
$$
r(\delta') \leq \alpha(d) \cdot (1-\gamma_{gap}),
$$
for some $\delta' \in 0.1(0.99 \delta,1.01\delta)$.
\end{lemma}

From this lemma, it is easy to show theorem (\ref{thm:robust}):
\begin{proof}(of theorem \ref{thm:robust})
The parameters of the theorem allow us to apply directly fact (\ref{fact:indset});
hence there exists a $k$-independent set $S$ of size at least $\eta n$, for $\eta$ as defined in fact(\ref{fact:indset}).
Hence, since $\delta \leq \eta/10$ there exists a $k$-independent set $S$ of size $10 \delta$.
Now, either:
\begin{enumerate}

\item
$S$ has expansion error at least $0.32$.
By lemma (\ref{lem:indexp}), for any $k\geq 4$ we have
$$
r(\mu) < \alpha(d)(1-\gamma_{gap}),
$$
for some $\mu \in 0.1 (0.99 \cdot 10 \delta,1.01 \cdot 10\delta) = (0.99 \delta,1.01\delta)$,
and $\gamma_{gap}(k)$ from lemma (\ref{lem:indexp}), 
which is at least $ min\left\{10^{-3},0.01/k \right\}$.

\item
The set $S$ is $\epsilon$-expanding for $\epsilon< 0.32$.  In which case,
since $S$ is in particular $L$-independent, then by corollary (\ref{cor:QECC}), 
the robustness $r(\delta')\le 2\eps < 2/3 - 0.01 \leq \alpha(d) -0.01$, for all $\delta' \leq |S|/n$.  In particular $r(\mu) < \alpha(1-0.01/k)$.

\end{enumerate}
Taking the higher of these two bounds we get the desired upper-bound for $r(\mu)$.
\end{proof}

\subsection{Proof of the lemma (\ref{lem:indexp})}

\subsubsection{General structure of proof} 
Let us clarify what we're trying to show.
We want to show that if the expansion is bad, errors cannot have large relative penalties.
Consider a set $S$ with positive expansion error $\eps>0$.
A-priori, if we have an error on $S$, then the maximal number of violations is strictly less than $|S| D_R$, and in fact at most $|S| D_R (1-\eps)$.
This might seem as though it proves the lemma trivially.

The technical problem here, however, is that 
an error on $S$ may just "seem" to be large, whereas possibly, may be represented much more succinctly modulo the stabilizer group.  
We would hence like to devise an error pattern, that cannot be downsized significantly, but would still "sense"
the non-expanding nature of $S$, and hence have fewer-than-optimal violations.
We start with a fact lower-bounding the weight modulo the group of an error confined to the qudits of a single generator; we call this the Onion 
fact since its proof (given in Subsection \ref{sec:onion})
works via some hybrid argument on
the onion-like layers $\Gamma^{(i)}(u)$ 
surrounding the qudits of one term $u$. 
\begin{fact}\label{fact:succinct}
\textbf{Onion fact}

\noindent
Let $C$ 
be a stabilizer code on $n$ qudits 
with a succinct generating set ${\cal G}$ of locality $k$, such that $dist(C)\geq k$.
Let ${\cal E}$ be some word in $\Pi^n$ s.t. $supp({\cal E}) 
\subseteq \Gamma^{(0)}(u)=\Gamma(u)$ for some generator $u\in {\cal G}$. 
Finally let $\Delta\in C$ be some word in the code and let
${\cal E}_{\cal G}=\Delta{\cal E}$. 
Then, for any $i\in [k]$, if 
$wt({\cal E}|_{\Gamma(u)}) = i$, then 
$wt({\cal E}_{\cal G}|_{\Gamma^{(k)}(u)}) \geq min\left\{i,k-i\right\}$.
\end{fact}

Now, let us see what the Onion fact means. 
It states that given an error on the $k$-qudit support of a generator, the weight of any
representation of this error modulo the centralizer, 
cannot be reduced in the $k$-neighborhood of the
generator, provided that the error has weight at most $k/2$.

Our idea is to concentrate the error on 
a large set of such "islands", 
each island supported on one generator, s.t. the generators 
supporting those islands are sufficiently far away from each other 
(in the interaction graph) 
so that the $k$-neighborhoods of those generators are non-intersecting.

If we draw a random error on the qudits on
these "islands", such that the expected number of errors per "island", 
is say $1$ error, the following will occur: 
on one hand, there will be a good portion of "islands" with at
least two errors, and these two errors will "sense" the 
sub-optimality of number of neighbors due to the expansion, 
interfere with each other, and cause an overall reduced penalty of 
that "island".
On the other hand, only a meager fraction, exponentially small in $k$, of those "islands" with at least two errors, will have more than $k/2$ errors; 
only those, by the Onion fact (fact \ref{fact:succinct}) 
can reduce their weight modulo the centralizer. 

We calibrate our parameters in the random choice of our error pattern 
to have expected number of errors which is indeed order of $1$, so that the 
above tradeoff will indeed hold. 

In the following we first define the error; 
We provide the proof that the expected penalty of this error is small in fact
(\ref{fact:penalty}), then prove the onion fact in Subsubsection 
\ref{sec:onion} and using it we then prove Fact (\ref{fact:weight}),
in which we show that the error has large weight modulo the group.
Finally we combine all the above to finish the proof of the lemma. 

\subsubsection{Constructing the error}
Let $U\subseteq L$ be a $k$-independent set as promised by the conditions of the lemma.
Then $|U| = \delta n$, and denoting $S=\Gamma(U)$, we have that $|S| = \delta n k$.
Therefore, $|\Gamma(S)| = |S| D_R (1-\eps')$, for some $\eps' \geq 0.32$.
Let ${\cal E}$ be the following random error process: for each qudit of $S$ independently, we apply $I$ w.p. $1-p$ for $p=1/(10k)$, and one of the other Pauli operators: with equal probability $p \cdot t(d)$, where $t$ is defined in (\ref{def:single}).

$$
{\cal E} = \bigotimes_{i\in S} {\cal E}_i
\mbox{,  where  }
{\cal E}_i =
\left\{
	\begin{array}{ll}
		I_i  & \mbox{w.p. } 1-1/(10k) \\
		X_d^k P_d^l & \mbox{w.p. } t/(10k)
	\end{array}
\right.
$$
We note here that the choice of $p$ is such that on average, each $k$-tuple has only a small number of errors; the expectation of the number of errors is an absolute constant $1/10$ (not a fraction of $k$). 
This will help, later on, to lower-bound the weight of the error modulo the group.

\subsubsection{Analyzing Penalty}
We first claim, that on average, ${\cal E}$ has a relatively 
small penalty w.r.t. ${\cal G}$, using the fact that the expansion error 
is at least $0.32$ as in the condition of Lemma \ref{lem:indexp}
\begin{fact}\label{fact:penalty}
$$
\mathbf{E}_{\cal E}\left[Penalty\right] 
\leq 
p \alpha |S| D_R \left(1- 0.02/k \right)
$$
\end{fact}

\begin{proof}
Let $G=(L,R;E)$ 
denote the bi-partite graph corresponding to ${\cal G}$, 
with $L$ being the constraints and $R$ the qudits. 
Let $S=\Gamma(U)$ be as before.
Let the error process ${\cal E}$ be the one defined above.
For any constraint $c\in \Gamma(S)$ which is violated when applied 
to this error, observe that there must be a qudit
$i\in supp(c)$ such that   
$\left[c|_i, {\cal E}_i\right] \neq 0$.  
We now would like to bound the number of constraints violated by ${\cal E}$ 
using this observation, and linearity of expectation. 

For an edge $e\in E$ connecting a qudit $i$ in $S$ and a constraint $c$ in 
$\Gamma(S)$,  
let $x(e)$ denote the binary variable which is $1$, iff the error 
term ${\cal E}_i$ on does not commute with $c|_i$.
In other words, an edge marked by $1$ is an edge whose qudit 
may cause its constraint to be violated. 
By construction, for each $e\in E$ which 
connects the qudit $i$ and the constraint $c$ we have 
\begin{equation}
\label{eq:singleexp}
\mathbf{E}_{\cal E}[x(e)] = p(1-t).
\end{equation} 
This is true since a constraint $c$ restricted to the qudit $i$, 
$c|_i$ does not commute with the error restricted to the same qudit $i$,  
${\cal E}_i$, iff both ${\cal E}_i$ is non-identity (which happens with 
probability $p$) and
is not equal to $c|_i$. 

If we had just added now $x(e)$ over all edges going out of $S$ 
(whose number is $|S| D_R$), then by linearity 
of expectation, this would have given an upper bound on the expected 
number of violated constraint equal to 

\begin{equation}\label{eq:sumexp}
\sum_e p(1-t)=p |S| D_R \alpha(d). 
\end{equation}

Unfortunately this upper bound does not suffice; to strengthen it 
we would now like to take advantage of the fact that 
many of those edges go to the same constraint, due to the fact that the 
expansion is bad; thus, instead of simply summing these expectation values, 
we take advantage of the fact that two qudits touching the same 
constraint
cannot contribute twice to its violation. 
We note that the bound we get does not gain back the full
factor of $1-\epsilon$ but a worse one. 

Let $E_{inj} \subseteq E$ be a subset of the edges between $S$ to $\Gamma(S)$ 
chosen by picking one edge for each constraint in $\Gamma(S)$. 
For an edge $e\in E$ let $c(e)$ denote the constraint incident on $e$, 
and let $e_{inj}(c(e))$ denote the edge
in $E_{inj}$ that is connected to $c(e)$.

We now bound the expectation by subtracting $x(e)$ from the sum, 
if the Boolean variable 
$x(e_{inj}(c(e)))$ is $1$; this avoids counting the violation of the same 
constraint twice due to the two edges. 
We have: 
$$
\mathbf{E}_{\cal E}\left[Penalty\right] 
\leq 
\mathbf{E}_{\cal E}\left[
\sum_{e\in E_{inj}} x(e) +
\sum_{e \notin E_{inj}} 
\left(1 - x(e_{inj}(c(e))) \right) \cdot x(e)
\right].
$$
Note that it may even be the case that some edges may cause 
constraints to become "unviolated", so the actual bound may be even lower.
Expanding the above by linearity of expectation:
$$
\mathbf{E}\left[Penalty\right] 
\leq 
\sum_{e\in E_{inj}} \mathbf{E}_{\cal E}\left[ x(e) \right] +
\sum_{e \notin E_{inj}} \mathbf{E}_{\cal E}\left[x(e)\right] - 
\sum_{e \notin E_{inj}} \mathbf{E}_{\cal E}\left[x(e_{inj}(c(e))) \cdot x(e) \right]=
$$
$$
\sum_{e\in E} \mathbf{E}_{\cal E}\left[ x(e) \right]+
\sum_{e \notin E_{inj}} \mathbf{E}_{\cal E}\left[x(e_{inj}(c(e))) \cdot x(e) \right]. 
$$

We have already calculated the first term in the sum in Equation 
\ref{eq:sumexp};  
We now lower bound the correction given by the second term. 
We use the fact that  for any $e\notin E_{inj}$
$$
\mathbf{E}_{\cal E}\left[x(e_{inj}(c(e)) x(e)\right] = 
\mathbf{E}_{\cal E}[x(e_{inj}(c(e)))]
 \mathbf{E}_{\cal E}[x(e)] $$
since ${\cal E}$ is independent between different qudits. 
We can thus substitute 
Equation \ref{eq:singleexp}, and get: 

$$
\mathbf{E}_{\cal E}\left[Penalty\right] 
\leq 
p \alpha |S| D_R  - |S| D_R \eps (p \alpha)^2.
$$
where we have used the fact that $|E_{inj}|=|S|D_R\eps$. 
This is equal to 
$$
p\alpha |S| D_R ( 1 - p\alpha \eps). 
$$
Using $p = 1/(10k),\eps \geq 0.32$, 
$\alpha(d) \geq 2/3$, we get the desired bound.
\end{proof}

\subsubsection{Towards analyzing the weight: Proof of the Onion fact} 
\label{sec:onion}

\begin{proof}{\bf (Of Fact \ref{fact:succinct})}
If $\Delta|_{\Gamma(u)}=I$ then 
\begin{equation}\label{eq:i} 
wt \left({\cal E}_{\cal G}|_{\Gamma^{(k)}(u)}\right) =
wt \left({\cal E}_{\cal G}|_{\Gamma(u)}\right) =
wt \left({\cal E}|_{\Gamma(u)}\right) = i.
\end{equation}
Otherwise, $\Delta|_{\Gamma(u)}$ 
is non-identity, and so has at least one non-identity 
coordinate, and also, since $\Delta$ is non-identity,  
by the assumption on the succinctness of ${\cal G}$ we 
have $wt(\Delta) \geq k$.

Moreover, we claim that $wt \left(\Delta|_{\Gamma^{(k)}(u)}\right) \geq k$.
Otherwise, consider the following process. 
Start with the constraint $u$, and consider the qudits in $\Gamma(u)=\Gamma^{(0)}(u)$. Now add the qudits belonging to all constraints in $\Gamma^{(1)}(u)$; 
Then add the next level, and so on until we have added add qudits belonging 
to $\Gamma^{(k)}(u)$. 
By the pigeonhole principle, if $wt \left(\Delta|_{\Gamma^{(k)}(u)}\right)<k$, 
then there must exist a level $t$ s.t. $1\le t \le k$ where  
$\Delta$ has zero support on qudits added in this level. 
This is in contradiction to the fact that the distance of the code 
is at least $k$, 
since we claim that ${\tilde \Delta}=\Delta|_{\Gamma^{(t-1)}(u)}$, 
is in the centralizer 
$\mathbf{Z}({\cal G})$
but its weight is less than $k$. 
To see that ${\tilde \Delta}$ is in the centralizer, 
we observe first that $\Delta$ commutes with all elements of ${\cal G}$ 
that act only on qudits in ${\Gamma^{(t-1)}(u)}$, 
and since ${\tilde \Delta}$ agrees with $\Delta$ on
${\Gamma^{(t-1)}(u)}$, ${\tilde \Delta}$ also commutes with them. 
We also observe that ${\tilde \Delta}$ trivially commutes with all 
elements in ${\cal G}$ whose support does not intersect 
$\Gamma^{(t-1)}(u)$. Hence we only need to worry about 
those terms that act on at least one qudit in $\Gamma^{(t)}(u)-\Gamma^{(t-1)}(u)$ 
and at least one qudit in  $\Gamma^{(t-1)}(u)$.  
Let $v$ be some such term. 
Note that $v$ cannot act on any qudit outside $\Gamma^{(t)}(u)$ 
by definition (of the $\Gamma^{(i)}(u)$'s).  
We know that $\Delta$ commutes with $v$.  
But by the choice of $t$, we know that $\Delta$ is trivial on those qudits 
added at the $t$th level, and hence $\Delta$ restricted to 
$\Gamma^{(t)}(u)$ (which contains the qudits of $v$)
is the same as $\Delta$ restricted to $\Gamma^{(t-1)}(u)$. 
And so $\Delta$ restricted to $\Gamma^{(t-1)}(u)$ commutes with $v$. 

We showed that ${\tilde \Delta}$ is in $\mathbf{Z}({\cal G})$.
If it also belongs to ${\cal G}$, this contradicts succinctness 
of ${\cal G}$;  
otherwise it is in $\mathbf{Z}({\cal G})-{\cal G}$ 
implying the distance of $C$ 
is at most $k-1$, contrary to assumption.
This means that $wt \left(\Delta|_{\Gamma^{(k)}(u)}\right) \geq k$.
Therefore, we now know that 
\begin{equation}\label{eq:k-i}
wt \left({\cal E}_{\cal G}|_{{\Gamma^{(k)}(u)}}\right) \geq 
wt \left(\Delta|_{{\Gamma^{(k)}(u)}}\right) - 
wt \left({\cal E}|_{{\Gamma^{(k)}(u)}}\right)
=
\end{equation}
$$
wt \left(\Delta|_{{\Gamma^{(k)}(u)}}\right) - 
wt \left({\cal E}|_{{\Gamma^{(0)}(u)}}\right)
\geq
k-i.
$$
Taking the minimal of the bounds from Equations (\ref{eq:i}),(\ref{eq:k-i})
 completes the proof.
\end{proof}

\subsubsection{Analyzing error weight}
We note that the expected weight of ${\cal E}$ is 
$p|S|$ and since $|S|$ is linear in $n$, 
by Chernoff  
the probability that the weight of ${\cal E}$ is smaller 
than by a constant fraction than this expectation is $2^{-\Omega(n)}$.
We need to show a similar bound on the weight modulo the centralizer group; 
given that $\delta<dist(C)/2$ we only need to bound the weight modulo 
the stabilizer group.  
Let $\Delta\in A$ be some element in the stabilizer group 
and let ${\cal E}_{\cal G}=\Delta{\cal E}$. 
We now need to lower-bound $wt({\cal E}_{\cal G})$.

\begin{fact}\label{fact:weight}
For integer $k$, let $\hat{k} = \floor{k/2}+1$.
Let $y(k): [4,\infty] \mapsto \mathbf{R}$ be the function:
$$
y(k) =
\left\{
	\begin{array}{ll}
		1-2^{(-\hat{k}+1)log(k)+k-2.3\hat{k}+4.54}  & k\geq 12 \\		
   	     0.9999 & 6\leq k\leq 11 \\
		0.9992 & k=5 \\
		0.9985 & k=4 
	\end{array}
\right.
$$
We claim:
$$
Prob_{\cal E} 
\left( 
wt({\cal E}_{\cal G}) < 
|S| p y(k)  
\right)
= 
2^{-\Omega(n)}.$$
\end{fact}

\begin{proof}
Let $x \sim B(k,p=1/(10k))$ denote a random variable which is the sum of $k$ 
i.i.d Boolean variables, each equal to $1$ with probability $p$; 
in other words, $x$ is a 
binomial process; $B(i) = Prob(x=i)$.
Let $U_i = \left\{u\in U | wt({\cal E}|_{\Gamma(u)}) = i \right\}$ be the 
set of constraints in which exactly $i$ errors occurred. 
Using the Hoeffding bound, for a given $i\in [k]$ and a given constant 
$\chi>0$, we have 
\begin{equation}
Prob_{\cal E} \left( \left|\frac{|U_i|}{|U|}-B(i)\right|\geq \chi\right) = 2^{-\Omega(n)}
\end{equation}
By the union bound, we have that for any constant $\chi>0$:
\begin{equation}\label{eq:qecc3}
Prob_{\cal E}\left(\exists i, s.t. \left|\frac{|U_i|}{|U|}-B(i)\right|\geq \chi\right) = 2^{-\Omega(n)}.
\end{equation}

Since the set $U$ is a $k$-independent set, then the sets $\left\{\Gamma^{(k)}(u)\right\}_{u\in U}$ are non-intersecting so
\begin{equation}\label{eq:qecc1}
wt({\cal E}_{\cal G}) \ge \sum_{u\in U} wt \left({\cal E}_{\cal G}|_{\Gamma^{(k)}(u)}\right),
\end{equation}
By the onion fact (\ref{fact:succinct}),
for each $u\in U_i$ we have
$wt \left({\cal E}_{\cal G}|_{\Gamma^{(k)}(u)}\right)\geq min\left\{i,k-i\right\}$, 
hence
$$
wt({\cal E}_{\cal G})
\ge
\sum_{i\in [k]}
|U_i| min\left\{i,k-i\right\}
=
\frac{|S|}{k}
\sum_{i\in [k]}
 \frac{|U_i|}{|U|}min\left\{i,k-i\right\}
$$
using $k|U|=|S|$. 

Therefore, using equation (\ref{eq:qecc3}) w.p. close to $1$ we have 

\begin{equation}\label{eq:qecc4}
wt({\cal E}_{\cal G})\ge \frac{|S|}{k} 
\sum_{i\in [k]} (B(i)- \chi) min \left\{i,k-i\right\} 
\geq 
\frac{|S|}{k}
\left(
\sum_{i\in [k]} B(i) min\left\{i,k-i\right\} - 2^{-k^2}
\right),
\end{equation}
for $\chi = 2^{-k^2}/k^2$.

We separate the rest of the proof to two cases: 
$k\ge 12$ and $4\le k< 12$.  
We start with the case $k\ge 12$. 
Recall $\hat{k} = \floor{k/2}+1$.  
Let
$$ A_{loss} = \sum_{i \geq \hat{k}} B(i) (2i - k).$$
Then by equation (\ref{eq:qecc4}) we have that with probability 
exponentially close to $1$ 
\begin{equation}\label{eq:qecc6}
wt({\cal E}_{\cal G}) 
\geq
\frac{|S|}{k}
\left( 
\sum_{i\in [k]} B(i) i
- 
A_{loss} - 
2^{-k^2} 
\right)
= 
\frac{|S|}{k} 
\left(
 p k - 
2^{-k^2} - 
A_{loss} 
\right)
\end{equation}

In the rest of the proof for $k\ge 12$ 
we upper-bound $A_{loss}$ and substitute 
in the above equation to derive the desired result.
Using an upper-bound of the binomial, we have: 
\begin{equation}\label{eq:binom}
B(\hat{k}) = 	{k \choose \hat{k}} p^{\hat{k}}(1-p)^{\hat{k}}
\leq 
2^{k} \cdot (10k)^{-\hat{k}} (1-p)^{\hat{k}}
\leq
k^{-\hat{k}} 10^{-\hat{k}}  2^k
\leq
2^{-\hat{k}log(k)+k-3.3\hat{k}},
\end{equation}
For any $i\geq \hat{k}$ and $p<1/2$ we have
\begin{equation}\label{eq:binbound}
B(i+1) = B(i) \left(\frac{k-i}{i+1}\right) \left(\frac{p}{1-p}\right) < B(i) \frac{p}{1-p} < 2 p B(i)
\end{equation}
Substituting equations (\ref{eq:binbound}) and (\ref{eq:binom}) in the expression for $A_{loss}$ we have:
\begin{equation}\label{eq:qecc5}
A_{loss} =
\sum_{i \geq \hat{k}}^k B(i) (2i-k)
\leq
2^{-\hat{k}log(k)+k-3.3\hat{k}} 
\sum_{i \geq {\hat{k}}}^k 
(2p)^{(i-\hat{k})} (2i-k)
\end{equation}
\begin{equation}
\leq
2^{-\hat{k}log(k)+k-3.3\hat{k}+1+\hat{k}} 
\sum_{i \geq {\hat{k}}}^k 
(p)^{(i-\hat{k})} (i-\floor{k/2})
\end{equation}
Changing summation $i - \floor{k/2} \mapsto j$ we have the above is at most:
\begin{equation}
2^{-\hat{k}log(k)+k-2.3\hat{k}+1} 
\sum_{j \geq 1}^{{\ceil{k/2}}} p^{-j+1} j
\leq
2^{-\hat{k}log(k)+k-2.3\hat{k}+1} 
\sum_{j \geq 1}^{{\ceil{k/2}}} p^{-j+1} k
\end{equation}
\begin{equation}
\leq
2^{-\hat{k}log(k)+k-2.3\hat{k}+1} 
 k 
\sum_{j \geq 1}^{{\ceil{k/2}}} p^{-j+1}
\leq
2^{-\hat{k}log(k)+k-2.3\hat{k}+1}
 k \cdot 1.1
\leq 
2^{(-\hat{k}+1)log(k)+k-2.3\hat{k}+1.2},
\end{equation}
where in the last inequality we bound the sum by $\sum_{i\geq 0} 1/p^i$,
and set $p = 1/(10k) \le 1/100$, using $k\ge 12$. 
Substituting this value in (\ref{eq:qecc6}) we have that with probability 
$2^{-\Omega(n)}$ close to $1$, 
$$
wt({\cal E}_{\cal G}) \geq 
\frac{|S|}{k}
\left( 
pk
- 
2^{-k^2} 
- 
2^{(-\hat{k}+1)log(k)+k-2.3\hat{k}+1.2} 
\right)
=
$$
$$
\geq
\frac{|S|}{k}
\left( 
pk
- 
2^{(-\hat{k}+1)log(k)+k-2.3\hat{k}+1.21} 
\right)
$$
where in the last inequality we used again $k\ge 12$. 
Continuing, using $p=\frac{1}{10k}$ the above bound is equal to 
$$
=|S| p \left(1 - 2^{(-\hat{k}+1)log(k)+k-2.3\hat{k}+1.21+log_2(10)}\right)
\geq 
|S| p y(k),
$$
for all $k\geq 12$.
For values of $4\leq k < 12$ we substitute directly $k$ in equation (\ref{eq:qecc4}), evaluate, and show it is at least $|S| p y(k)$.

\end{proof}

\subsubsection{Concluding the proof of lemma (\ref{lem:indexp})}
\begin{proof}
By fact (\ref{fact:penalty}) the average penalty of ${\cal E}$ is small on average, i.e. 
$$\mathbf{E}\left[Penalty({\cal E})\right] \leq |S| D_R p \alpha (1-0.02/k) 
\triangleq P.$$
Yet, by fact (\ref{fact:weight}) w.p. exponentially close to $1$, we have 
$$
wt({\cal E}_{\cal G}) 
\geq 
|S| p y(k) \triangleq W_{low}\geq |S| p \cdot 0.99.
$$
Similarly, by the Hoeffding bound w.p. exponentially close to $1$, we have
$$
wt({\cal E}_{\cal G}) < |S| p (1+0.01) \triangleq W_{high}.
$$

Since all penalties are non-negative, we conclude that {\it conditioned} on
$\left|wt({\cal E}_{\cal G})/ (|S|p) -1\right| < 0.01$, we have
$\mathbf{E}\left[Penalty({\cal E})\right] \leq P+2^{-\Omega(n)}$.
Therefore, there must exist an error ${\cal E}$, whose weight modulo
 ${\cal G}$ deviates
by a fraction at most $0.01$ from $|S| p$, and whose penalty is at most $P+2^{-\Omega(n)}$.

We would like to bound the robustness of this error, which is the ratio 
of the penalty to its relative weight times $D_R$. 
We get that its robustness is at most 
\begin{equation}\label{eq:robust}
r = \frac{P+2^{-\Omega(n)}}{D_R W_{low}}\leq 
\frac{1}{D_R}
\cdot
\frac
{|S| D_R p \alpha (1-0.019/k)}
{|S| p y(k)}
=
\alpha \left(\frac{1-0.019/k}{y(k)}\right).
\end{equation}

We now note that in the last expression,  for all $k\geq 12$, 
the ratio $\frac{1-0.019/k}{y(k)}$ 
is at most $1-0.01/k$. 
For all values of $4 \leq k < 12$ we substitute the appropriate 
value of $y(k)$ and get similarly that 
the ratio $\frac{1-0.019/k}{y(k)}$ is at most $1- 10^{-3}$. 
Hence, the robustness of the error, $r$ is at most $\alpha(d)(1-\gamma_{gap})$ 
where $\gamma_{gap}$ is as defined in the statement of Theorem 
\ref{thm:robust}. 

\end{proof}

\section{Acknowledgements}
The authors would like to thank Gil Kalai for useful discussions regarding 
expander graphs; and Eli Ben-Sasson and Irit Dinur for discussions about 
PCPPs and LTCs.

\section{Appendix} 
For $S\subseteq R$ let $\Gamma_1(S)\subseteq \Gamma(S)$ denote the subset of the neighbors of $S$ with exactly one neighbor in $S$.
Similarly, let $\Gamma_{\geq 2}(S)$ denote
the subset of neighbors with at least two neighbors in $S$.

\paragraph{Proof of fact(\ref{fact:essence}):}
\begin{proof}
The average degree of a vertex in $\Gamma(S)$ w.r.t. $|S|$ is at most $\frac{|S|D_R}  { |S| D_R (1-\eps)} = \frac{1}{1-\eps}$.
Let $\alpha_1$ denote the fraction $|\Gamma_1(S)| / |\Gamma(S)|$.
Then
$$ \frac{1}{1-\eps} = \alpha_1 1 + (1-\alpha_1) m,$$
where $m$ is the average degree of a vertex with at least two neighbors in $S$.
Then 
$$ \alpha_1 = 1-\frac{\eps m}{m-1}. $$
Since $m\geq 2$, then $\alpha_1$ is minimized for $m=2$, and therefore
$$ \alpha_1 \geq 1-2\eps.$$
\end{proof}

\paragraph{Proof of fact (\ref{fact:deg}):}
\begin{proof}
By definition, we have $|\Gamma(S)| \geq |S| D_R (1-\eps)$.
Let $E_{inj}\subseteq E(S)$ be a subset of the edges incident on $S$ such that each $u\in \Gamma(S)$ has 
a single neighbor in $S$ connected by an edge of $E_{inj}$.
Then $E_{inj}$ is of size $\Gamma(S)$ which is at least $|S| D_R (1-\eps)$.
Also $|E(S)| = |S| D_R$, thus $|E(S)-E_{inj}| \leq |S| D_R \eps$.
Therefore $\left| \Gamma_{\geq 2} (S) \right| \leq |S| D_R \eps$.
Hence, $\Gamma_1(S) = \Gamma(S)-\Gamma_{\geq 2}(S)$ is of size at least $|S| D_R (1-\eps) - |S| D_R \eps = |S| D_R (1-2\eps)$.
Therefore, there exists a vertex $v\in S$ with at least $D_R (1-2\eps)$ neighbors in $\Gamma_1(S)$.
Since $v$ has $D_R$ neighbors in $\Gamma(S)$, then the fraction of neighbors of $v$ with at least two neighbors in $S$ is at most $2\eps$.
\end{proof}

\paragraph{Proof of Claim (\ref{cl:classical}):} 
\begin{proof}
The construction of $\cite{CRVW}$, generates explicitly
for any $\eps,r$ a right-regular bi-partite 
graph $G=(L,R;E)$ whose right degree is $D_R$ such that $|L|/|R|=1-r$,  
and for any subset $S\subseteq R$, $|S|\leq |R| \delta$ the neighbor set of $S$ is of size at least $|S| D_R (1-\eps)$, where $D_R$ is the right degree of $G$.
Note that since the right degree is $D_R$, the average left degree is 
$D_R|R|/|L|=D_R\frac{1}{1-r}$, which is a constant given that $D_R$ is 
a constant. 

The code is defined by assigning to each left node a 
parity check over its incident vertices.  
Let us lower bound the rate of this code: it is at least 
$r=(|R|-|L|)/|R|$, since each constraint in $L$ removes one dimension 
of the space. The minimal distance of the code is at least $\delta$, 
since any non-zero word of weight at most $\delta$ is rejected, 
since there exists at least one check term
that "sees" just a single bit at state $1$, by Fact (\ref{fact:essence}). 

Hence, these are so-called "good" codes. 
Furthermore, their robustness is at least $1-3\eps$ since 
an error on a set of bits $S$ of size $|S|<\delta n$, 
is examined by at least $|S|D_R(1-\epsilon)$ constraints. 
By Fact (\ref{fact:essence}) at least $1-2\epsilon$ of those constraints, examine $S$ 
in exactly one location; all 
constraints that touch a given error set $S$ in exactly one location 
will be violated; hence the total number of constraints that will be 
violated is at least $|S| D_R (1-\eps) (1-2\eps)\geq |S| D_R (1-3\eps)$. 
\end{proof}

\end{document}